\documentclass[a4paper]{article}

\usepackage{graf_GR}
\usepackage[T1]{fontenc}
\usepackage{authblk}

\usepackage{slashed}

\newcommand{\Deth}{\De_\varth}

\renewcommand{\Re}{\mathrm{Re}}
\renewcommand{\Im}{\mathrm{Im}}

\title{Mode stability results for the Teukolsky equations \\ on Kerr-anti-de Sitter spacetimes}

\author[1]{Olivier Graf\thanks{ograf@uni-muenster.de}}
\author[1,2]{Gustav Holzegel\thanks{gholzegel@uni-muenster.de}}

\affil[1]{\small Westf\"alische Wilhelms-Universit\"at M\"unster,
Mathematisches~Institut, Einsteinstrasse~62~48149~M\"unster,~Bundesrepublik~Deutschland \vskip.2pc \ }
\affil[2]{\small Imperial College London,
Department of Mathematics,
South~Kensington~Campus,~London~SW7~2AZ,~United~Kingdom}

\begin{document}
\maketitle
\begin{abstract}
  We prove that there are no non-stationary (with respect to the Hawking vector field), real mode solutions to the Teukolsky equations on all $(3+1)$-dimensional subextremal Kerr-anti-de Sitter spacetimes. We further prove that stationary solutions do not exist if the black hole parameters satisfy the Hawking-Reall bound and $\left|a\sqrt{-\Lambda}\right|<\frac{\sqrt{3}}{20}$. We conclude with the statement of mode stability which preludes boundedness and decay estimates for general solutions which will be proven in a separate paper. Our boundary conditions are the standard ones which follow from fixing the conformal class of the metric at infinity and lead to a coupling of the two Teukolsky equations. The proof relies on combining the Teukolsky-Starobinsky identities with the coupled boundary conditions. In the stationary case the proof exploits elliptic estimates which fail if the Hawking-Reall bound is violated. This is consistent with the superradiant instabilities expected in that regime.  
\end{abstract}

\section{Introduction and main results}\label{sec:intro}
The \emph{Kerr-anti-de Sitter spacetimes} (Kerr-adS) are solutions to the \emph{Einstein equations}\footnote{In this article we set the speed of light $c$ and the gravitational constant $G$ to $1$.}
\begin{align}\label{eq:EEcosmo}
  \RRRic(g) -\half \mathrm{R}(g) + \Lambda g = 0,
\end{align}
with negative cosmological constant $\La=:-3k^2<0$, given by
\begin{align*}
  \begin{aligned}
    g_{\mathrm{KadS}} & := -\frac{\Delta}{\Xi^2\Si}\le(\d t - a \sin^2\varth\d\varphi\ri)^2 + \frac{\Si}{\Delta}\d r^2 + \frac{\Si}{\Delta_\varth}\d\varth^2 + \frac{\Delta_\varth}{\Xi^2\Si}\sin^2\varth\le(a\d t - (r^2+a^2)\d\varphi\ri)^2,
  \end{aligned}
\end{align*}
where
\begin{align*}
  \Si & := r^2+a^2\cos^2\varth, & \Delta & := (r^2+a^2)\le(1+k^2r^2\ri) - 2Mr, \\
  \Delta_\varth & := 1 - a^2k^2\cos^2\varth, &  \Xi & := 1-a^2k^2.
\end{align*}
The following parameters define Kerr-adS spacetimes containing a non-degenerate subextremal black hole region.
\begin{definition}\label{def:admissibleKadSparameters}
  The parameters $(M,a,k)$ belong to the set of \emph{admissible subextremal Kerr-adS black hole parameters} if\,\footnote{Up to making the change of variable $\varphi\to-\varphi$, we can always assume that $a\geq0$. The constant $k$ can be chosen to be positive by definition.} 
  \begin{align*}
    M,k & >0, & a & \geq 0, & ak & < 1,
  \end{align*}
  and if the set $\{r\geq 0~:~\De(r)=0\}$ has two distinct elements.\footnote{This last requirement generalises the well-known sub-extremality condition $|a|<M$ in the asymptotically flat $k=0$ case.} The maximum element $r_+(M,a,k)>0$ is called the \emph{black hole radius}.
\end{definition}

We consider the metric $g_{\mathrm{KadS}}$ with parameters as above on the manifold $\MM_{\mathrm{KadS}} :=  \RRR_t\times(r_+,+\infty)_r\times\SSS^2_{\varth,\varphi}$. It is well-known that $(\MM_{\mathrm{KadS}},g_{\mathrm{KadS}})$ can be smoothly extended to a larger manifold containing a black hole region and with $\MM_{\mathrm{KadS}}$ embedding isometrically as its \emph{domain of outer communications}. In particular, we can attach two null boundaries $\mathcal{H}^{-}$ and $\HH^+$ to $\MM_{\mathrm{KadS}}$, such that $\mathcal{H}^-\cup\mathcal{H}^+=\{r=r_+\}$, and which are known as the past and future \emph{event horizon}. See Figure \ref{fig:integrationtransport} below.\\



Unlike their $\Lambda>0$ and $\Lambda=0$ counterparts, whose exterior stability properties have been intensely studied over the past two decades~\cite{Vas13,Dya16,Hin.Vas18,Sch16,Mav21,Mav21a,Fan21,Fan22} and~\cite{Daf.Hol.Rod19,And.Bac.Blu.Ma19, Haf.Hin.Vas21,Daf.Hol.Rod.Tay21, Kla.Sze20,Kla.Sze21}, the stability properties of Kerr-adS spacetimes have remained more elusive. A first distinct feature of the $\Lambda<0$ case is that stability is to be understood in the context of an \emph{initial boundary value problem} for the Einstein equations, for which \emph{boundary conditions} have to be imposed at the conformal infinity of the spacetime (see~\cite{Fri95,Enc.Kam19}). It is expected that stability will strongly depend on the choice of boundary conditions (see~\cite{Hol.Luk.Smu.War20}). A second distinct feature is that the nature of the (in)stability properties of the Kerr-adS black hole will crucially be tied to conditions on its parameters (the mass $M$, the specific angular momentum $a$ and the cosmological scale $k$).\\

This rich dynamical behaviour can already be illustrated in the context of the so-called ``toy-stability problem'', which is the problem of understanding the asymptotic properties of solutions to the conformal covariant wave equation on Kerr-adS\footnote{Around the Kerr-adS solutions, the Einstein equations~\eqref{eq:EEcosmo} are approximated by conformal wave-type equations called \emph{Teukolsky equations}. A first model problem to understand the stability of the Kerr-adS solutions is to understand the conformal wave equation~\eqref{eq:confwave} which is a simplification of the Teukolsky equations. See Section~\ref{sec:introTeuk} for further discussions.}
\begin{align}
  \label{eq:confwave}
  \Box_{g_{\mathrm{KadS}}} \psi + 2k^2\psi =0.
\end{align}
For Kerr-adS spacetimes satisfying the so-called \emph{Hawking-Reall bound}~\cite{Haw.Rea99}
\begin{align}
  \label{est:HRboundoriginal}
  a \leq k r_+^2,
\end{align}
all solutions with Dirichlet boundary conditions decay inverse logarithmically in time~\cite{Hol.Smu13} and this decay rate is sharp for general solutions \cite{Hol.Smu14} (see also~\cite{Gan14}). This \emph{weak decay} is due to a stable trapping mechanism caused by the reflecting Dirichlet boundary conditions. On the other hand, if the Hawking-Reall bound is violated, \cite{Dol17} established the existence of periodic and exponentially growing solutions. Geometrically, the Hawking-Reall bound~\eqref{est:HRboundoriginal} guarantees the existence of a globally causal Killing vector field, the absence of which allows for so-called \emph{superradiant instabilities}.\footnote{Superradiant instabilities are absent in the asymptotically flat case but appear for instance for the \emph{massive} wave equation on Kerr even with $\Lambda=0$, the mass acting as a confinement mechanism just as the reflecting Dirichlet boundary conditions do in the Kerr-adS case (see~\cite{Shl14}). For the full linearised Einstein equations on Kerr-adS, the existence of superradiant instabilities have been obtained by a virial-type argument in~\cite{Gre.Hol.Ish.Wal16} using the so-called \emph{canonical energy}.} Note that if the Dirichlet conditions are replaced by optimally dissipative boundary conditions, strong inverse polynomial decay is expected in the Hawking-Reall regime (see~\cite{Hol.Luk.Smu.War20} in the anti-de Sitter case $M=a=0$). The exponential growth or the (too) slow decay of solutions with Dirichlet conditions suggested the following instability conjecture.\footnote{The (in)stability of the anti-de Sitter limiting space $M=a=0$ with Dirichlet-type conditions is still an open question. See~\cite{Biz.Ros11} for numerical simulations revealing a turbulent instability mechanism in the context of the Einstein-scalar field system,~\cite{Mos20,Mos23} for recent breakthrough results establishing the instability of anti-de Sitter space in the context of the spherically symmetric Einstein-Vlasov system, and~\cite{Cha.Smu24} for the construction of periodic solutions to non-linear wave equations on anti-de Sitter space.}
\begin{conjecture*}[\cite{Hol.Smu13}]
  The Kerr-anti-de Sitter spaces are (non-linearly!) unstable solutions to the initial boundary value problem for Einstein equations~\eqref{eq:EEcosmo} with Dirichlet-type boundary conditions.  
\end{conjecture*}

The present paper initiates a program to address the above conjecture and extend the results for the toy-stability problem to the actual Einstein equations near the Kerr-adS family. Our first objective is to establish the slow logarithmic decay rate for the linear problem to provide further support to the above conjecture: \\

\emph{Prove that Kerr-adS spacetimes satisfying the Hawking-Reall bound~\eqref{est:HRboundoriginal} are linearly stable solutions of the Einstein equations~\eqref{eq:EEcosmo} -- \emph{i.e.} perturbations decay inverse logarithmic in time (and generally not better) to a linearised Kerr-adS spacetime -- if Dirichlet-type conditions are imposed on the linear perturbation of the metric.}\\

The main theorems of the present paper (see Theorem~\ref{thm:main} and its generalisation Theorem~\ref{thm:fullmode}) are a prelude to the above objective: we prove mode stability results for the so-called \emph{Teukolsky equations} on Kerr-adS spacetimes. The rest of the introduction (Section~\ref{sec:intro}) is organised as follows.
\begin{itemize}
\item In Section~\ref{sec:introTeuk}, we introduce the Teukolsky equations and discuss their relevancy for the study of the asymptotic stability of Kerr-adS spacetimes. We set up the initial boundary value problem associated the Teukolsky equations on Kerr-adS spacetimes (see Theorem~\ref{thm:wellposedness}).
\item In Section~\ref{sec:intromodestab}, we introduce the notion of \emph{mode} solutions to the Teukolsky equations and discuss their relations to linear stability. We provide a full version of the main theorem of this paper (Theorem~\ref{thm:main}), which ensures that there does not exist mode solutions to the Teukolsky equations with \emph{real} frequencies, except for a potential set of specific frequencies which depends on the Kerr-adS black hole parameters. The proof of Theorem~\ref{thm:main} constitutes the core of this paper (Sections~\ref{sec:radang} to~\ref{sec:stat}).
\item In Section~\ref{sec:introfullmode}, we use the real mode stability of Theorem~\ref{thm:main} to prove a \emph{full} mode stability result (see Theorem~\ref{thm:fullmode}). This relies on a continuity argument together with the full mode stability of the Teukolsky equations on specific Schwarzschild-anti-de Sitter spacetimes -- \emph{i.e.} when $a=0$ --, see Theorem~\ref{thm:boundedSadS}. The proof of that result uses the \emph{Chandrasekhar transformation theory} between solutions of the Teukolsky equations and solutions of the so-called \emph{Regge-Wheeler} equations.
\item In Section~\ref{sec:overview} we provide an overview of the proof of Theorem~\ref{thm:main} and highlight the differences with the proof of the corresponding mode stability results in the asymptotically flat, Kerr case.
\item Finally, in the last Sections~\ref{sec:related},~\ref{sec:finalcomments}, we discuss related results in the mathematical and physics literature and how we will use the results of this paper for our program. 
\end{itemize}
The bulk of the paper (Sections~\ref{sec:radang} to~\ref{sec:stat}) is devoted to the proof of Theorem~\ref{thm:main} and is organised as follows.
\begin{itemize}
\item Section~\ref{sec:radang} is dedicated to the separability properties of the Teukolsky equations into ODEs.
\item In Section~\ref{sec:TStrafos} we introduce the so-called \emph{Teukolsky-Starobinksy transformations} which are symmetries of the respective separated Teukolsky ODEs.
\item In Section~\ref{sec:TSconslaws}, we use the Teukolsky-Starobinsky transformations to derive \emph{conservation laws}. These conservation laws rule out the existence of \emph{non-stationary} mode solutions to the Teukolsky equations.
\item In Section~\ref{sec:stat}, we further rule out the existence of stationary mode solutions using elliptic estimates for the ODEs.
\item Finally, in Appendix~\ref{sec:Robin}, we provide an alternative proof of the non-existence of non-stationary modes obtained in Section~\ref{sec:TSconslaws}. This uses the derivation of Robin boundary conditions for mode solutions to the the Teukolsky initial boundary value problem.
\end{itemize}


\subsection{The Teukolsky equations}\label{sec:introTeuk}
Unlike the toy-stability problem, the problem of linear stability for the Einstein equations immediately forces one to address the question of gauge, \emph{i.e.}~which coordinate system the Einstein equations should be linearised in. However, it is also well-known -- in complete analogy with the asymptotically flat case -- that there is a large \emph{gauge independent aspect} of linear stability that can be addressed. From the work of Teukolsky~\cite{Teu72} we know that two of the linearised null curvature components\footnote{The quantities $\al^{[+2]}, \al^{[-2]}$ are obtained from null components of the Riemann tensor defined with respect to the so-called \emph{algebraically special null frame}. They correspond to the respective horizontal tensors $\al,\alb$ in the Christodoulou-Klainerman framework and to the respective quantities $\Psi_0$ and $(r-ia\cos\varth)^4\Psi_4$ in the Newman-Penrose  formalism, where $\Psi_0,\Psi_4$ are defined with respect to the so-called \emph{Kinnersley tetrad}.} $\al^{[+2]}$ and $\al^{[-2]}$ are gauge invariant \emph{to linear order} and satisfy the following decoupled wave equations (see~\cite{Kha83} for their expression in the general $\La\neq0$ case)
\begin{subequations}\label{Teusys}
  \begin{align}\label{eq:Teukoriginal}
    \begin{aligned}
      0 & = \Box_{g_{\mathrm{KadS}}} \al^{[\pm2]} \pm 2\frac{\pr_r\De}{\Si}\pr_r\al^{[\pm2]} \pm 2\Xi\frac{\pr_r\De}{\Si\De}\le(a\pr_\varphi+(r^2+a^2)\pr_t\ri)\al^{[\pm2]} \\
        & \quad \mp 8\frac{\Xi r}{\Si}\pr_t\al^{[\pm2]} + \frac{\le(6k^2r^2 + (1\pm1) \pr^2_r\De\ri)}{\Si}\al^{[\pm2]} \\
        &  \quad \pm 4i\frac{\Xi\cot\varth}{\Si\De_\varth}(1-a^2k^2\cos(2\varth))\le(a\sin\varth\pr_t+\frac{\pr_\varphi}{\sin\varth}\ri)\al^{[\pm2]} \\
        & \quad \mp 8i\frac{\Xi}{\Si} a \cos\varth\pr_t\al^{[\pm2]} -\frac1\Si\le(2\De_\varth+2a^2k^2+4\frac{\Xi^2\cot^2\varth}{\De_\varth}\ri)\al^{[\pm2]}.
    \end{aligned}
  \end{align}
  The Teukolsky quantities $\al^{[\pm2]}$ are \emph{spin-weighted complex functions}, for which we give the following compact \emph{ad hoc} definition (see~\cite[Section 2.2.1]{Daf.Hol.Rod19a}). We refer the reader to~\cite{Daf.Hol.Rod19a,Mil24} for more geometric definitions. 
\begin{definition}\label{def:smoothspin}
  Let $s\in\RRR$. A \emph{smooth spin-$s$-weighted complex function on $\MM_{\mathrm{KadS}}$} is a function \\ $\al:\RRR_t\times(r_+,+\infty)_r\times(0,\pi)_\varth\times (2\pi\mathbb{T}^1)_\varphi \to \CCC$ such that $\al$ is smooth  on $\RRR_t\times(r_+,+\infty)_r\times (0,\pi)_\varth\times (2\pi\mathbb{T}^1)_\varphi$ and such that $e^{is\varphi}(\widetilde{Z}_1)^{k_1}(\widetilde{Z}_2)^{k_2}(\widetilde{Z}_3)^{k_3} \al$ and $e^{-is\varphi}(\widetilde{Z}_1)^{k_1}(\widetilde{Z}_2)^{k_2}(\widetilde{Z}_3)^{k_3} \al$ extend continuously at $\varth=0$ and $\varth=\pi$ respectively for all $k_1,k_2,k_3\in\mathbb{N}$, where
  \begin{align*}
    \widetilde{Z}_1 & := -\sin\varphi\pr_\varth + \cos\varphi\le(-is\csc\varth-\cot\varth\pr_\varphi\ri), & \widetilde{Z}_3 & := \pr_\varphi,\\
    \widetilde{Z}_2 & := -\cos\varphi\pr_\varth - \sin\varphi\le(-is\csc\varth-\cot\varth\pr_\varphi\ri).
  \end{align*}
\end{definition}

Another question to be addressed is what ``Dirichlet conditions'' should mean for the system of linearised Einstein equations. The most natural and fully geometric ``Dirichlet conditions'' for the Einstein equations in the context of the initial boundary value problem of~\cite{Fri95,Enc.Kam19} is to impose that the conformal class of the induced metric on the boundary at infinity coincides with that of the anti-de Sitter metric. This directly translates into the following boundary conditions for the null curvature components $\alpha^{[\pm 2]}$ 
\begin{align}
  \label{eq:defbdycond}
  \widetilde{\al}^{[+2]} - \big(\widetilde{\al}^{[-2]}\big)^\ast  & \xrightarrow{r\to+\infty} 0, & r^2\pr_r\widetilde{\al}^{[+2]} + \big(r^2\pr_r\widetilde{\al}^{[-2]}\big)^\ast & \xrightarrow{r\to+\infty} 0,
\end{align}
where $^\ast$ denotes the complex conjugation, and where $\widetilde{\al}^{[\pm2]}$ are the following renormalisations
\begin{align*}
  \widetilde{\al}^{[+2]} & := \De^2(r^2+a^2)^{-3/2}\al^{[+2]}, & \widetilde{\al}^{[-2]} & := (r^2+a^2)^{-3/2}\al^{[-2]}. 
\end{align*}
We call~\eqref{eq:defbdycond} the \emph{conformal anti-de Sitter boundary conditions} for the Teukolsky quantities, and we emphasise that they \emph{couple} the dynamical variables $\al^{[+2]}$ and $\al^{[-2]}$ at the boundary. This situation differs from the asymptotically flat case where each spin-$\pm2$ quantity can be studied independently. We refer the reader to~\cite[Section 6]{Hol.Luk.Smu.War20} for the details on the derivation of~\eqref{eq:defbdycond}. Moreover, we point out that the boundary conditions~\eqref{eq:defbdycond} are also the most natural ones in the context of metric perturbation using Hertz maps~\cite{Dia.Rea.San09,Dia.San13}, which is expected given the duality between the two approaches~\cite{Wal78}.\\

One usually restricts to study metric perturbations which are \emph{regular at the event horizon}. For the Teukolsky quantities $\al^{[\pm2]}$ this imposes the following regularity conditions.
\begin{definition}\label{def:reghor}
  Let $\al^{[\pm 2]}$ be two smooth spin-$\pm2$-weighted quantities. Let $r^\ast$ be the \emph{tortoise coordinate} given by
  \begin{align*}
    \frac{\d r^\ast}{\d r} & = \frac{r^2+a^2}{\De}, & r^\ast(r=+\infty) & = \frac{\pi}{2},
  \end{align*}
  and let $t^\ast$ be the regular time function given by\,\footnote{Any other regular time function is equally acceptable. See also the presentation in~\cite[Section 2]{Hol.Smu13}.}
  \begin{align*}
    t^\ast & := t + r^\ast - \frac{1}{k}\arctan(k r).
  \end{align*}
  We say that $\al^{[\pm 2]}$ are \emph{regular at the event horizon in the future of $t_0^\ast \in \RRR $} if, for all $t^{\ast}_1\geq t^\ast_0$, we have\footnote{The regularity assumptions~\eqref{eq:defreghor} can be deduced from the fact that the vector fields $L,\De^{-1}\Lb$ are regular at the horizon and that regular metric perturbations give rise to regular spacetime curvature tensors.}
  \begin{align}
    \label{eq:defreghor}
    L^p\le(\De^{-1}\Lb\ri)^q\le(\widetilde{\al}^{[+2]}\ri)\bigg|_{t^\ast=t^\ast_1} & = O_{r\to r_+}(1), & L^p\le(\De^{-1}\Lb\ri)^q\le(\De^{-2}\widetilde{\al}^{[-2]}\ri)\bigg|_{t^\ast=t^{\ast}_1} & = O_{r\to r_+}(1),
  \end{align}
  for all $p,q\in\mathbb{N}$, 
  and where the pair of null vector fields $L,\Lb$ is defined by
  \begin{align*}
    L & := \Xi\pr_t + \frac{a\Xi}{r^2+a^2}\pr_\varphi + \pr_{r^\ast}, & \Lb & := \Xi\pr_t + \frac{a\Xi}{r^2+a^2}\pr_\varphi - \pr_{r^\ast}.  
  \end{align*}
\end{definition}
\end{subequations}

Under these conditions, the Teukolsky equations give rise to a well-posed initial boundary value problem.
\begin{theorem}\label{thm:wellposedness}
  Let $t_0^\ast\in\RRR$ and let $\al^{[\pm2]}_0,\al^{[\pm2]}_1$ be (smooth) initial data quantities on $\Si_{t_0^\ast}:=\{t^\ast=t^\ast_0\}$, compatible with the boundary conditions at infinity~\eqref{eq:defbdycond} and with the regularity conditions at the horizon~\eqref{eq:defreghor}. Then, there exist two (smooth) regular quantities $\al^{[\pm2]}$ defined on the future $J^+(\Si_{t_0^\ast})$ of $\Si_{t_0^\ast}$ in $\MM_{\mathrm{KadS}}$ such that $\al^{[\pm2]}|_{\Si_{t_0^\ast}}=\al_0^{[\pm2]}$ and $\pr_t\al^{[\pm2]}|_{\Si_{t_0^\ast}} = \al_1^{[\pm2]}$, which solve the Teukolsky equations~\eqref{eq:Teukoriginal} and which satisfy the conformal anti-de Sitter boundary conditions at infinity~\eqref{eq:defbdycond} and the regularity conditions at the horizon in the future of $t_0^\ast$~\eqref{eq:defreghor}.\footnote{See Figure~\ref{fig:integrationtransport} for a Penrose diagram depicting the setting of Theorem~\ref{thm:wellposedness}.}
\end{theorem}
\begin{proof}
  Defining
  \begin{align*}
    \al^D & := \big(\De^2\al^{[+2]}\big)-\big(\al^{[-2]}\big)^\ast, & \al^N & :=\big(\De^2\al^{[+2]}\big)+\big(\al^{[-2]}\big)^\ast,
  \end{align*}
  the system~\eqref{Teusys} falls in the framework of strong hyperbolic operators with homogeneous (respectively Dirichlet and Neumann) boundary conditions for the unknowns $(\al^{D},\al^N)$ as defined in~\cite{War15}. The result then follows from Theorem 2.3 in~\cite{War15}.  
\end{proof}

\begin{figure}[h!]
  \centering
  \includegraphics[height=0.3\textheight]{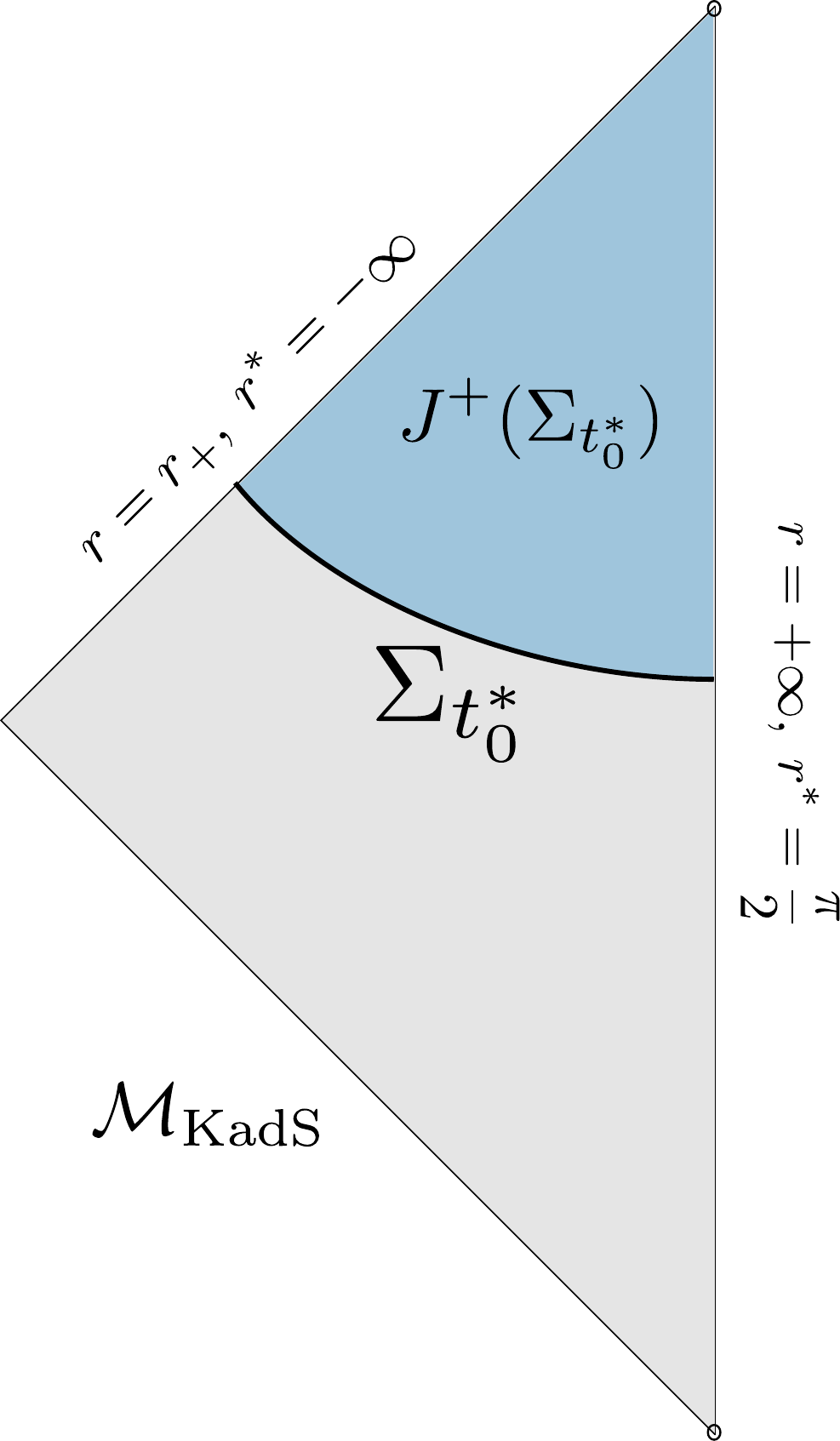}
  \caption{The spacetime $\MM_{\mathrm{KadS}}$ and the initial boundary value problem of Theorem~\ref{thm:wellposedness}.}
  \label{fig:integrationtransport}
\end{figure}

Theorem~\ref{thm:wellposedness} tells us that the Teukolsky system~\eqref{Teusys} can be studied independently of having to address the problem of gauge. In fact, it is even believed that the Teukolsky quantities control the full solutions to the Einstein linearised equations up to pure gauge solutions and linearised Kerr(-adS) solutions. See~\cite{Wal73} where it was shown, in the $\La=0$ case, that if $\al^{[\pm2]}$ vanish identically, then the full solution is trivial (\emph{i.e.} it is a sum of a pure gauge solution and a linearised Kerr solution), and~\cite{Daf.Hol.Rod19} where, in the Schwarzschild case, quantitative decay of the full solution (up to pure gauge solutions and linearised Kerr) is obtained from the decay of $\al^{[\pm2]}$. With this in mind, in the Kerr-adS case, a cornerstone of our objective mentioned before Section~\ref{sec:introTeuk} is the following:\\

\emph{Show that the solutions to the Teukolsky system~\eqref{Teusys} decay inverse logarithmically in time provided that the background Kerr-adS spacetime satisfies the Hawking-Reall bound~\eqref{est:HRboundoriginal}.}

\subsection{Mode stability on the real axis and the main theorem}\label{sec:intromodestab}
In this paper, we prove \emph{mode stability} for the initial boundary value problem of Theorem~\ref{thm:wellposedness}. Mode stability -- in conjunction with the theory of scattering resonances\footnote{See~\cite{Dya.Zwo19} for a general introduction to this theory.} -- has been a fundamental ingredient to obtain the decay of solutions to wave-type equations: see for example~\cite{Haf.Hin.Vas21} for the proof of the linear stability of Kerr with slow rotation $|a|\ll M$, or~\cite{Hin.Vas18,Fan21,Fan22} for the proof of the linear and non-linear stability of the slowly rotating Kerr-de Sitter solutions, which all rely on resonance expansions and the identification of obstructions to mode stability. In another framework, \emph{quantitative} versions of mode stability (see~\cite{Shl15,Tei20}) were used in the proof of the decay of solutions to the wave and Teukolsky equations on Kerr in the full subextremal range $|a|<M$ in~\cite{Daf.Rod.Shl16} and~\cite{Shl.Tei20}.\\

For the Teukolsky equations~\eqref{Teusys}, mode solutions take the following form.
\begin{definition}\label{def:mode}
  Let $\om\in\CCC$. We say that $\al^{[\pm2]}$ are (smooth) \emph{modes with frequency $\om$} if there exists (smooth) time-independent functions $\hat\al^{[\pm2]}$ such that
  \begin{align}\label{modesol}
    \alpha^{[+2]} (t,r,\varth,\varphi) & = e^{-i\omega t} \hat{\alpha}^{[+2]} (r,\varth,\varphi), & \alpha^{[-2]}(t,r,\vartheta,\varphi) & = e^{+i\omega^\star t} \hat{\alpha}^{[-2]} (r,\varth,\varphi),
  \end{align}
  for all $(t,r,\varth,\varphi)\in\RRR\times(r_+,+\infty)\times(0,\pi)\times(0,2\pi)$.
\end{definition}
\begin{remark}
  In the literature, Teukolsky mode quantities $\al^{[\pm2]}$ are usually defined with the convention $\al^{[\pm2]} = e^{-i\om t}\hat\al^{[\pm2]}$. While any convention is equally acceptable when $\al^{[+2]}$ and $\al^{[-2]}$ are not coupled, the coupling of these two quantities \emph{via} the conformal anti-de Sitter boundary condition~\eqref{eq:defbdycond} yields that the only convention potentially leading to the existence of non-trivial modes and fitting into the framework of~\cite{War15} is the one of Definition~\ref{def:mode}.
\end{remark}
\begin{remark}
  For solutions to the full linearised Einstein equations, there can be no non-trivial mode solution with the convention of Definition~\ref{def:mode} alone, due to the coupling of $\al^{[+2]}$ and $\al^{[-2]}$ \emph{via} the Bianchi equations (see~\cite{New.Pen62} and the Teukolsky-Starobinsky constraint relations~\cite{Sta.Chu74,Teu.Pre74}). In that case the correct definition of mode solution is sums of complex-conjugated frequencies mode solutions in the sense of Definition~\ref{def:mode}, \emph{i.e.}
  \begin{align*}
    \al^{[+2]} & = e^{- i\om t}\hat\al^{[+2]}_{-} + e^{+ i\om^\ast t}\hat\al^{[+2]}_+, & \al^{[-2]} & = e^{+ i\om^\ast t}\hat\al^{[-2]}_{-} + e^{- i\om t}\hat\al^{[-2]}_+,
  \end{align*}
  where
  \begin{align*}
    \al^{[+2]}_- & := e^{-i\om t}\hat\al^{[+2]}_{-}, & \al^{[-2]}_- & := e^{+i\om^\ast t}\hat\al^{[-2]}_{-},\\
    \intertext{and}
    \al^{[+2]}_+ & := e^{+ i\om^\ast t}\hat\al^{[+2]}_+, & \al^{[-2]}_+ & := e^{- i\om t}\hat\al^{[-2]}_+, 
  \end{align*}
  are respectively mode solutions to the Teukolsky equations in the sense of Definition~\ref{def:mode}.
\end{remark}

Mode stability is the statement that there do not exist regular modes which are solutions of the initial boundary value problem for Teukolsky equations with $\mathrm{Im}(\omega) \geq 0$. This article is dedicated to the proof of the following theorem which establishes \emph{mode stability on the real axis}.
\begin{theorem}[Main theorem]\label{thm:main}
  Let $(M,a,k)$ be admissible subextremal Kerr-adS black hole parameters in the sense of Definition~\ref{def:admissibleKadSparameters}. Then the following holds.
  \begin{enumerate}
  \item There exist no non-zero smooth modes with real frequency $\om\in\RRR$ in the sense of Definition~\ref{def:mode} which are solutions of the Teukolsky equations~\eqref{eq:Teukoriginal}, which satisfy the conformal anti-de Sitter boundary conditions at infinity~\eqref{eq:defbdycond} and the regularity conditions at the horizon~\eqref{eq:defreghor}, \underline{except} for a potential set of \underline{stationary modes} with respect to the \emph{Hawking vector field}, i.e.~solutions satisfying
  \begin{align*}
    \mathrm{K}(\al^{[\pm2]}) & = 0, & \mathrm{K} & := \pr_t+\frac{a}{r_+^2+a^2}\pr_\varphi.
  \end{align*}
  \item The exceptional non-trivial stationary solutions from (1) cannot exist if one of the following holds:
  
\begin{enumerate}
\item the solution is also time-independent, \emph{i.e.} $\pr_t\al^{[\pm2]}=0$. 
\item the Hawking-Reall bound~\eqref{est:HRboundoriginal} is satisfied, and one of the following two \emph{ad hoc} hypotheses hold:
  \begin{subequations}\label{est:slowrotmore}
    \begin{align}
      ak & \leq \frac{1}{20},\label{est:assumboundX}\\
      \intertext{or}
      a & \leq \frac{1}{4}kr_+^2 \label{est:farHR} \, .
    \end{align}
  \end{subequations}
  \end{enumerate}
 \end{enumerate}
\end{theorem}
\begin{remark}
  Projecting the potentially non-trivial stationary mode solutions of Theorem~\ref{thm:main} onto the Hilbert basis $(e^{im\varphi})_{m\in\ZZZ}$ of the functions which are $2\pi$-periodic in $\varphi$, we obtain directly that their associated frequencies must be of the form  
  \begin{align*}
    \om & = m \om_+,
  \end{align*}
  with $m \in\ZZZ\setminus\{0\}$, and where $\om_+:=a/(r_+^2+a^2)$. 
\end{remark}
\begin{remark}
  In view of the numerical results of~\cite{Car.Dia.Har.Leh.San14}, we expect that the second part of the theorem actually holds without the assumptions~\eqref{est:slowrotmore}. However, in view of the bounds obtained in the present paper, the numerical results of~\cite{Car.Dia.Har.Leh.San14} and the superradiant instabilities revealed in~\cite{Gre.Hol.Ish.Wal16}, we expect that our result is sharp with respect to the Hawking-Reall bound, \emph{i.e.} that mode stability does not hold if it is violated. Note that this violation can happen even under the additional assumption~\eqref{est:assumboundX}. See the discussions in Section~\ref{sec:related} and Remark~\ref{rem:sharpHR}. 
\end{remark}

Theorem~\ref{thm:main} is the main result of this paper. After some preliminaries in Sections \ref{sec:radang} and \ref{sec:TStrafos}, the first part is proven in Section \ref{sec:TSconslaws} and the second part in Section \ref{sec:stat}. We provide in Section~\ref{sec:overview} a brief overview of its proof.

\subsection{From mode stability on the real axis to full mode stability}\label{sec:introfullmode}
It is well-known that one can exploit the continuity of the mode frequencies in the black hole parameters to deduce the full statement of mode stability for \eqref{Teusys} from mode stability on the real axis together with mode stability for specific values of the parameters, e.g.~the Schwarzschild case $a=0$ (see~\cite{And.Ma.Pag.Whi17}).\footnote{The general idea is that modes ``have to cross the real axis'' to become exponentially growing modes.}\\

For a large class of Schwarzschild-adS spacetimes, we can obtain mode stability for $\Im(\om)>0$ (and more general uniform boundedness statements) relatively easily by generalising the so-called \emph{Chandrasekhar transformation} theory of \cite{Daf.Hol.Rod19a,Daf.Hol.Rod19} in the asymptotically flat case. We recall that the main difficulty to obtain uniform boundedness statements for \eqref{Teusys}, even in the case $a=0$, is that there is no natural conserved energy or integrated local energy decay estimate for solutions. Using the Chandrasekhar transformation, one obtains for $a=0$ the following \emph{Regge-Wheeler} system
\begin{subequations}\label{Teusys2}
  \begin{align}
    0 & = L \underline{L} \Psi^{[\pm 2]} + \frac{\Delta}{r^4} \le(\LL-\frac{6M}{r}\ri) \Psi^{[\pm 2]},
  \end{align}
  and
  \begin{align}\label{eq:RWBC}
    \begin{aligned}
      \Psi^{[+2]} - \big(\Psi^{[-2]}\big)^\ast & \xrightarrow{r\to +\infty} 0, \\
      r^2\pr_r\Psi^{[+2]} +r^2\pr_r\big(\Psi^{[-2]}\big)^\ast & \xrightarrow{r\to +\infty} -6M\lim_{r\to+\infty}\le(\widetilde{\al}^{[+2]}+\big(\widetilde{\al}^{[-2]}\big)^\ast\ri), 
    \end{aligned}
  \end{align}
\end{subequations}
where
\begin{align}\label{eq:defPsipm2}
  \Psi^{[+2]} & := \frac{r^4}{\De}\Lb\le(\frac{r^4}{\De}\Lb\widetilde{\al}^{[+2]}\ri), & \Psi^{[-2]} & := \frac{r^4}{\De}L\le(\frac{r^4}{\De}L\widetilde{\al}^{[-2]}\ri), 
\end{align}
and where $\LL$ is an angular Laplace-type elliptic operator with eigenvalues $\ell(\ell+1)$ for $\ell\geq 2$. Defining
\begin{align}\label{eq:defPsiDR}
  \Psi^D & := \Psi^{[+2]} - \big(\Psi^{[-2]}\big)^\ast, & \Psi^R & := \LL(\LL-2)\le(\Psi^{[+2]} + \big(\Psi^{[-2]}\big)^\ast\ri) +12M\pr_t\le(\Psi^{[+2]} - \big(\Psi^{[-2]}\big)^\ast\ri),
\end{align}
one can show that~\eqref{Teusys} and~\eqref{Teusys2} imply that 
\begin{subequations}\label{Teusys2bis}
  \begin{align}
    0 & = L \underline{L} \Psi + \frac{\Delta}{r^4} \le(\LL-\frac{6M}{r}\ri) \Psi,
  \end{align}
  for $\Psi = \Psi^D,\Psi^R$, and
  \begin{align}\label{eq:RWBC}
    \begin{aligned}
      \Psi^D & \xrightarrow{r\to +\infty} 0, \\
      2\pr_{t}^2\Psi^R + \frac{\LL(\LL-2)}{6M}\pr_{r^\ast}\Psi^R + k^2\LL\Psi^R & \xrightarrow{r\to+\infty} 0. 
    \end{aligned}
  \end{align}
\end{subequations}

Integrating by parts and exploiting the boundary conditions for the system~\eqref{Teusys2bis} gives the \emph{boundedness in time} of the energies
\begin{align}\label{eq:energySadS}
  \begin{aligned}
     \mathrm{E}_{m\ell}^D(t^\ast) &:= \half\int_{\Si_{t^\ast}} \frac{r^2(1+k^2r^2)^2-4M^2}{r^2(1+k^2r^2)}|\pr_{t^\ast}\Psi^{D}_{m\ell}|^2 + |\pr_{r^\ast}\Psi^{D}_{m\ell}|^2 + \frac{\De}{r^4}\le(\ell(\ell+1) - \frac{6M}{r}\ri) |\Psi^{D}_{m\ell}|^2 \,\d r^\ast, \\
    \mathrm{E}_{m\ell}^R(t^\ast) & := \half\int_{\Si_{t^\ast}} \frac{r^2(1+k^2r^2)^2-4M^2}{r^2(1+k^2r^2)}|\pr_{t^\ast}\Psi^{R}_{m\ell}|^2 + |\pr_{r^\ast}\Psi^{R}_{m\ell}|^2 + \frac{\De}{r^4}\le(\ell(\ell+1) - \frac{6M}{r}\ri) |\Psi^{R}_{m\ell}|^2 \,\d r^\ast \\
    & \quad + \frac{6M}{\ell(\ell+1)\le(\ell(\ell+1)-2\ri)} \le[|\pr_{t^\ast}\Psi^R_{m\ell}|^2 + k^2\frac{\ell(\ell+1)}{2}|\Psi^R_{m\ell}|^2\ri]\bigg|_{r=+\infty},
  \end{aligned}
\end{align}
where $\Psi_{m\ell}$ denotes the projections of $\Psi$ on the eigenfunctions of $\LL$ and where $\pr_{t^\ast},\pr_{r^\ast}$ are the coordinate vector fields associated to the $(t^\ast,r^\ast)$ coordinates.
\begin{remark}\label{rem:coercivity}
  It is immediate that since $\ell\geq 2$, the energies~\eqref{eq:energySadS} are \underline{coercive provided the parameters are such} \underline{that $r_+ \geq M$}. We highlight that for Schwarzschild-adS spacetimes the bound $r_+\geq M$ does not always hold. Indeed, for fixed mass $M$, the black hole radius $r_+$ can range from $0$ (when $k\to+\infty$) to the Schwarzschild radius $2M$ (when $k\to0$). 
  
\end{remark}
From this observation we deduce the following theorem.

\begin{theorem}\label{thm:boundedSadS}
  For Schwarzschild-adS spacetimes with parameters satisfying $r_+ \geq M$, there does not exist non-zero, regular, growing mode $\Im(\om)>0$ solutions to the Teukolsky system~\eqref{Teusys}. 
\end{theorem}
\begin{proof}
  Let $\al^{[\pm2]}$ be regular growing mode $\Im(\om)>0$ solutions to the Teukolsky system~\eqref{Teusys}. The transformations $\Psi^{D},\Psi^R$ given by~\eqref{eq:defPsipm2},~\eqref{eq:defPsiDR} produce regular growing mode solutions (perhaps identically zero) to the Regge-Wheeler system~\eqref{Teusys2bis}. This contradicts the boundedness and coercivity of the energies~\eqref{eq:energySadS}, except if $\Psi^{D}=\Psi^R=0$. From~\eqref{eq:defPsiDR}, this implies that $\Psi^{[+2]}=\Psi^{[-2]}=0$. We now have the following general transport estimates (see Figure~\ref{fig:integrationtransportrevised})
    \begin{align}
      \norm{\psi}_{L^\infty\le(J^+(\Si_{t^\ast_0})\ri)} & \les \big\Vert\frac{r^4}{\De}L\psi\big\Vert_{L^\infty\le(J^+(\Si_{t^\ast_0})\ri)} + \norm{\psi}_{L^\infty(\Si_{t^\ast_0})},\label{est:transportpsi}
    \end{align}
    for all functions $\psi$ on $J^+(\Si_{t^\ast_0})$.\\
    
    \begin{figure}[h!]
      \centering
      \includegraphics[height=0.35\textheight]{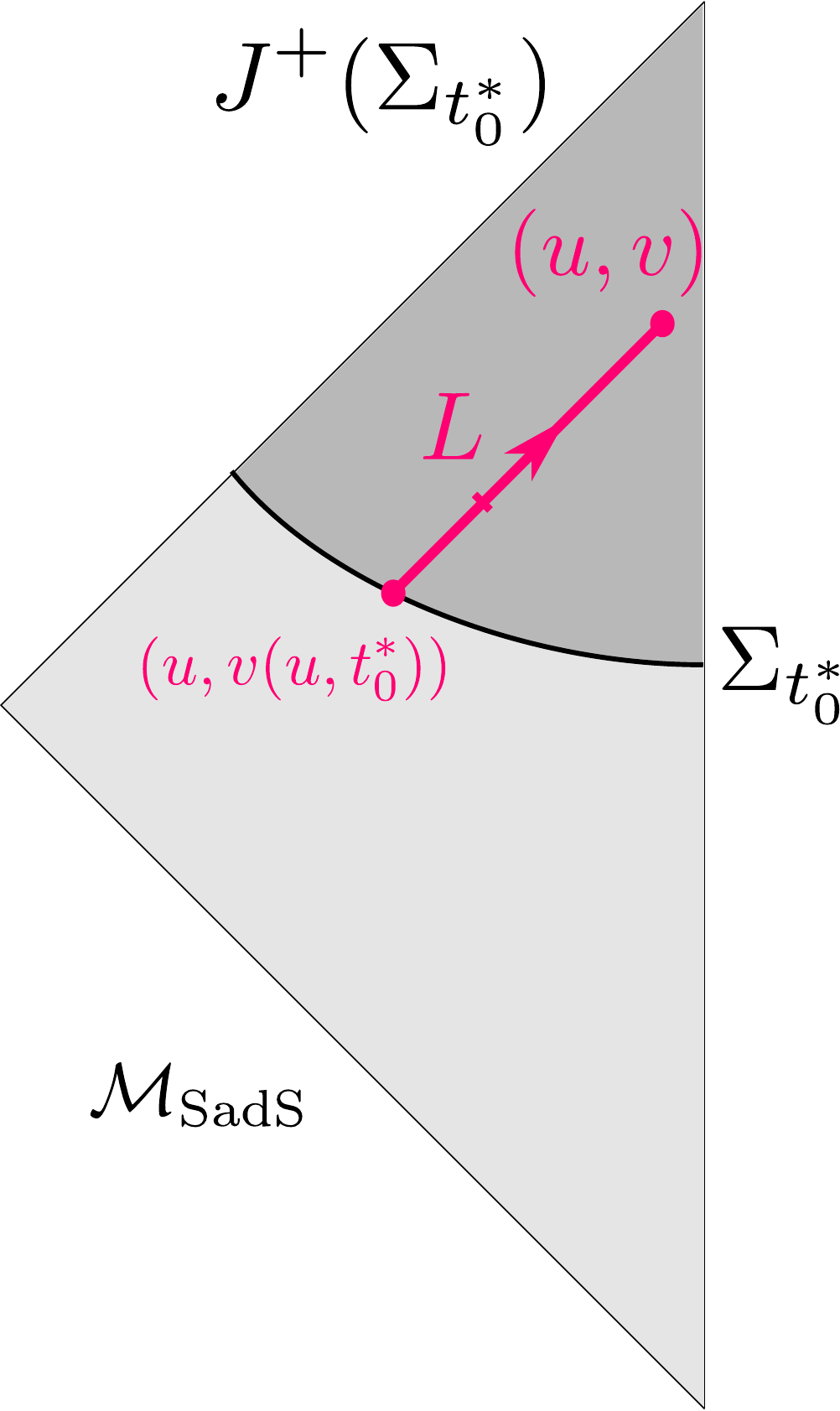}
      \caption{The transport estimate~\eqref{est:transportpsi}.}
      \label{fig:integrationtransportrevised}
    \end{figure}
    Applying estimate~\eqref{est:transportpsi} once, using the vanishing of $\Psi^{[-2]}$ and the definition of the Chandrasekhar transformations~\eqref{eq:defPsipm2}, we deduce that
    \begin{align}\label{est:intermediateboundednessRW}
      \big\Vert\frac{r^4}{\De}L\widetilde{\al}^{[-2]}\big\Vert_{L^\infty\le(J^+(\Si_{t^\ast_0})\ri)} \les \big\Vert\frac{r^4}{\De}L\widetilde{\al}^{[-2]}\big\Vert_{L^\infty(\Si_{t^\ast_0})} < \infty.
    \end{align}
  Applying estimate~\eqref{est:transportpsi} again, using~\eqref{est:intermediateboundednessRW}, we deduce that $\widetilde{\al}^{[-2]}$ is bounded. This contradicts the fact that $\al^{[-2]}$ is a growing mode, except if $\al^{[-2]}=0$. From the boundary conditions~\eqref{eq:defbdycond}, this implies that $\al^{[+2]}$ satisfies both Dirichlet and Neumann boundary conditions. From a standard uniqueness argument, we deduce that $\al^{[+2]}=0$.
\end{proof}

\begin{remark}
  Elaborating on the proof of Theorem~\ref{thm:boundedSadS}, one can more generally obtain a quantitative uniform boundedness statement for solutions to the Teukolsky system~\eqref{Teusys} on Schwarzschild-adS spacetimes with parameters satisfying $r_+\geq M$. 
\end{remark}

Unfortunately, the above argument is not as simple in the case of general Schwarzschild-adS black holes -- in view of the potentially negative potential in~\eqref{Teusys2} --, or in the case of non-vanishing angular momentum -- in which case the equations for $\Psi^{[\pm 2]}$ involve a coupling to $\alpha^{[\pm 2]}$. Nonetheless, the continuity of mode frequencies together with the real axis mode stability of Theorem~\ref{thm:main} enables to improve the mode stability result of Theorem~\ref{thm:boundedSadS} to the following full mode stability statement.
\begin{theorem}\label{thm:fullmode}
  Let $(M,a,k)$ be admissible subextremal Kerr-adS black hole parameters such that the Hawking-Reall bound~\eqref{est:HRboundoriginal} and the additional hypothesis~\eqref{est:slowrotmore} hold. Then there do not exist non-zero mode solutions to the Teukolsky system~\eqref{Teusys} with $\Im(\om) \geq 0$. 
\end{theorem}
\begin{proof}
  The Teukolsky system~\eqref{Teusys} belongs to the class of strongly hyperbolic operator with homogeneous boundary conditions considered in~\cite{War15}. The Laplace-transformed Teukolsky operator and its dual are therefore locally uniformly Fredholm, see~\cite[Theorems 4.3 and 4.6, Equations (4.4) and (4.11)]{War15}. Using the associated estimates one shows directly that the set of frequencies $\om\in\CCC$ associated to non-trivial modes is continuous in the parameters $(M,a,k)$ (see also \cite[Proposition 5.11]{Hin.Vas18} and~\cite[Paragraph 2.7]{Vas13} where these arguments are carried out in greater generality). \\
  Assume that there exists a mode solution with $\Im(\om)>0$ for admissible parameters $(M,a,k)$ satisfying the Hawking-Reall bound and the additional hypothesis. There exists a smooth path of such parameters joining $(M,a,k)$ to the unit mass Schwarzschild-adS parameters $(1,0,1)$. Since $r_+(1,0,1)=1$, by Theorem~\ref{thm:boundedSadS}, there exists no mode with $\Im(\om)>0$ for these parameters. Therefore, by continuity of the frequency $\om\in\CCC$, there exists admissible parameters satisfying the Hawking-Reall bound and the additional hypothesis such that the Teukolsky equations admit a real mode $\Im(\om)=0$. This contradicts Theorem~\ref{thm:main} and finishes the proof of the result.
\end{proof}

\subsection{Overview of the proof and comparison with the asymptotically flat case}\label{sec:overview}
We provide below a brief sketch of the proof of Theorem~\ref{thm:main}. The first steps are common to the asymptotically flat and Kerr-adS cases. First, we decompose regular real mode solutions $\widetilde{\alpha}^{[\pm 2]}$ as
\begin{align*}
  \widetilde{\alpha}^{[+2]} (t,r,\vartheta,\varphi) & = e^{-i\omega t} \sum_{m\in\ZZZ}\sum_{\ell\geq|m|} R_{[+2],\omega}^{m \ell}(r) S_{m\ell}^\omega(\vartheta) e^{+im\varphi},\\
  \widetilde{\alpha}^{[-2]}(t,r,\vartheta,\varphi) & = e^{+i\omega t} \sum_{m\in\ZZZ}\sum_{\ell\geq|m|} R_{[-2],\omega}^{m \ell}(r) S_{m\ell}^\omega(\vartheta) e^{-im\varphi}.
\end{align*}
From the separability properties of the Teukolsky equations~\eqref{eq:Teukoriginal}, the radial functions $R_{[+2]}$ and $R_{[-2]}$ satisfy two second order ODEs called \emph{radial Teukolsky equations}. Each solution of these equations can be decomposed in a basis of regular/irregular solutions at the horizon:
\begin{align*}
  R_{[\pm2]}(r) & = A_{[\pm2],\HH^+}R_{[\pm2],\HH^+}(r) + A_{[\pm2],\HH^-}R_{[\pm2],\HH^-}(r). 
\end{align*}
Note that the regularity conditions at the future horizon~\eqref{eq:defreghor} for $\widetilde{\al}^{[\pm2]}$ impose that the radial Teukolsky quantities $R_{[\pm2]}$ are on the regular branches of the ODE at the horizon, \emph{i.e.}
\begin{align}\label{eq:radregintro}
  A_{[+2],\HH^{-}} = A_{[-2],\HH^-} = 0.
\end{align}
Let us first consider the non-stationary case $\om-m\om_+\neq0$. The analysis for the asymptotically flat and asymptotically anti-de Sitter case separates in two distinct cases.
\begin{itemize}
\item {\bf The asymptotically flat case.} In the asymptotically flat case, each solution decomposes on a basis of regular/irregular solutions at infinity:
  \begin{align*}
    R_{[\pm2]}(r) & = A_{[\pm2],\II^+}R_{[\pm2],\II^+}(r) + A_{[\pm2],\II^-}R_{[\pm2],\II^-}(r).
  \end{align*}
  Regularity conditions at infinity will impose that
  \begin{align}\label{eq:radregintrobis}
    A_{[+2],\II^-} = A_{[-2],\II^-} = 0.
  \end{align}
  Using the so-called \emph{Teukolsky-Starobinsky identities} -- which is an invariance of the Teukolsky ODEs by some fourth order differential operators (see Section~\ref{sec:TStrafos} for their definitions in the general Kerr-adS case) -- and a conservation of the Wronskian argument, one can obtain the following two \emph{Teukolsky conservation laws}
  \begin{subequations}\label{eq:conslawintroflat}
  \begin{align}
    \begin{aligned}
      & (\om-m\om_+)^{-1}\frac{\wp^{-1}}{\widetilde{C}(r_+^2+a^2)}|A_{[+2],\HH^+}|^2 + 16\om^5|A_{[+2],\II^+}|^2 \\
      = \;\; & (\om-m\om_+)\frac{|\xi|^2\widetilde{C}}{r_+^2+a^2}|A_{[+2],\HH^-}|^2 + \frac{\wp}{16\om^3} |A_{[+2],\II^-}|^2,
    \end{aligned}
  \end{align}
  and
  \begin{align}
    \begin{aligned}
      & (\om-m\om_+)\frac{|\xi|^2\widetilde{C}}{r_+^2+a^2}|A_{[-2],\HH^+}|^2 +\frac{\wp}{16\om^3}|A_{[-2],\II^+}|^2 \\
      = \;\; &(\om-m\om_+)^{-1}\frac{\wp^{-1}}{\widetilde{C}(r_+^2+a^2)}|A_{[-2],\HH^-}|^2 + 16\om^5|A_{[-2],\II^-}|^2,   
    \end{aligned}
  \end{align}
  \end{subequations}
  where $\wp$ is the so-called \emph{radial Teukolsky-Starobinsky constant} (see Lemma~\ref{lem:radTS} for a definition in the more general Kerr-adS case), and where $\widetilde{C}$ and $|\xi|^2$ are -- for subextremal black hole parameters -- strictly positive constants. Note that the radial Teukolsky-Starobinsky constant $\wp$ can also easily be shown to be strictly positive. The conservation laws~\eqref{eq:conslawintroflat} and the regularity conditions~\eqref{eq:radregintro} and~\eqref{eq:radregintrobis} imply that
  \begin{align}\label{eq:conslawintroflatbis}
    \begin{aligned}
      (\om-m\om_+)^{-1}\frac{\wp^{-1}}{\widetilde{C}(r_+^2+a^2)}|A_{[+2],\HH^+}|^2 + 16\om^5|A_{[+2],\II^+}|^2 & = 0, \\
      (\om-m\om_+)\frac{|\xi|^2\widetilde{C}}{r_+^2+a^2}|A_{[-2],\HH^+}|^2 +\frac{\wp}{16\om^3}|A_{[-2],\II^+}|^2 & = 0. 
    \end{aligned}
  \end{align}
  From~\eqref{eq:conslawintroflatbis} and the previous remarks, we deduce that in the \emph{non-superradiant case} $\om(\om-m\om_+)>0$,
  \begin{align*}
    A_{[+2],\HH^+} = A_{[+2],\II^+} = A_{[-2],\HH^+} = A_{[-2],\II^+} = 0,
  \end{align*}
  which proves mode stability. The above approach fails in the superradiant case $\om(\om-m\om_+)<0$. In that case, the proof of mode stability was obtained using a transformation theory introduced in~\cite{Whi89} (initially for growing modes $\Im(\om)>0$), generalised to real modes for the wave equation in~\cite{Shl15} and for the Teukolsky equations in~\cite{And.Ma.Pag.Whi17}, and to include the extremal case in~\cite{Tei20}. Recently, the transformation of Whiting~\cite{Whi89} has been given a new profound interpretation in terms of spectral symmetries in~\cite{Cas.Tei22}.\footnote{The spectral symmetries highlighted in~\cite{Cas.Tei22} also allowed the authors to provide partial mode stability results in the Kerr-de Sitter case $\La>0$. See also~\cite{Hin21} where mode stability results are obtained for a specific class of such spacetimes. The general $\Lambda>0$ case remains an open problem.}
\item {\bf The Kerr-adS case.} In the Kerr-adS case, the two radial Teukolsky quantities $R_{[+2]}$ and $R_{[-2]}$ couple at infinity \emph{via} the conformal anti-de Sitter boundary conditions~\eqref{eq:defbdycond} as:
  \begin{align}\label{eq:bdycondintro}
    R_{[+2]}-R_{[-2]}^\ast & \xrightarrow{r\to+\infty} 0, & r^2\pr_rR_{[+2]}+r^2\pr_rR_{[-2]}^\ast & \xrightarrow{r\to+\infty} 0.
  \end{align}
  Using Teukolsky-Starobinsky identities and conservation of Wronskians coupled at infinity by the boundary conditions~\eqref{eq:bdycondintro}, we obtain a Teukolsky conservation law at the horizon:
  \begin{align}\label{eq:conslawintroadS}
    \begin{aligned}
      & \aleph |A_{[+2],\HH^+}|^2 + \le|\xi\Xi(\om-m\om_+)\widetilde{C} A_{[-2],\HH^+}-12M(\Xi\om)A_{[+2],\HH^+}^\ast\ri|^2 \\
      = \;\;& \aleph |A_{[-2],\HH^-}|^2 + \le|\xi\Xi(\om-m\om_+)\widetilde{C} A_{[+2],\HH^-}-12M(\Xi\om)A_{[-2],\HH^-}^\ast\ri|^2,
    \end{aligned}
  \end{align}
  where $\aleph$ is the so-called \emph{angular Teukolsky-Starobinsky constant} (see Lemma~\ref{lem:angTS} for a definition), and where $\widetilde{C}$ and $\xi$ are the non-vanishing, Kerr-adS versions of the corresponding constants of the asymptotically flat case. Note that the angular Teukolsky-Starobinsky constant $\aleph$ can be shown to be strictly positive (see Lemma~\ref{lem:angTS} in the present paper). The conservation laws~\eqref{eq:conslawintroadS} and the regularity conditions at the horizon~\eqref{eq:radregintro} imply that
  \begin{align}\label{eq:conslawintroadSbis}
    \aleph |A_{[+2],\HH^+}|^2 + \le|\xi\Xi(\om-m\om_+)\widetilde{C} A_{[-2],\HH^+}-12M(\Xi\om)A_{[+2],\HH^+}^\ast\ri|^2 = 0.
  \end{align}
  From~\eqref{eq:conslawintroadSbis} and the above remarks, we deduce that
  \begin{align*}
    A_{[+2],\HH^+} = A_{[-2],\HH^+} = 0,
  \end{align*}
  which proves mode stability in the non-stationary case.
\end{itemize}
Let us make the following two comments.
\begin{itemize}
\item In the asymptotically flat case, the superradiant obstruction to mode stability is caused by the comparison between the values of Wronskians \underline{at the horizon} and \underline{at infinity}. Indeed, while the asymptotics at the horizon produce $(\om-m\om_+)$ factors, the asymptotics at infinity produce $\om$ factors. This leaves room for the possibility that these two factors do not have the same signs, and that not all terms in the conservation laws~\eqref{eq:conslawintroflatbis} are positive. In the Kerr-adS case, using the ``reflection at infinity'' given by~\eqref{eq:bdycondintro} between the spin-$+2$ and the spin-$-2$ Teukolsky quantities, we only compare values of Wronskians \underline{at the horizon}. The computations of the asymptotics thus only produce $(\om-m\om_+)$ terms and leave no room to sign discrepancies, thus eluding any superradiant obstruction. We emphasise that the above argument excludes non-trivial mode solutions with $\omega-m\omega_+ \neq 0$ \underline{for all} admissible Kerr-adS subextremal black hole parameters, \emph{i.e.} \underline{even if the Hawking-Reall bound~\eqref{est:HRboundoriginal} is violated}.
\item It is remarkable that it is the \underline{angular} Teukolsky-Starobinsky constant which appears in the conservation law in the Kerr-adS case~\eqref{eq:conslawintroadS}, and not the \underline{radial} Teukolsky-Starobinsky constant as in the asymptotically flat case~\eqref{eq:conslawintroflat}. The strict positivity of the angular Teukolsky-Starobinsky constant, which is crucial to the argument, requires a delicate analysis (see Lemma~\ref{lem:angTS}). Note that the radial Teukolsky-Starobinsky constant can be expressed as $\wp = \aleph + 144M^2(\Xi\om)^2$ for which it is simpler to establish the strict positivity required in the asymptotically flat case.
\end{itemize}

The second part of Theorem~\ref{thm:main} concerns the stationary case $\om-m\om_+=0$. It is obtained by angular and radial elliptic estimates (see Section~\ref{sec:stat}). The cornerstone is the positivity of a potential, which is linked to the Hawking-Reall bound (see Remark~\ref{rem:sharpHR}), and is obtained in this paper using additionally the hypothesis~\eqref{est:slowrotmore}. This is where the superradiant obstructions to mode stability arise, which, as mentioned above, are absent from the non-stationary case.

\subsection{Mathematical and numerical results on Kerr-adS spacetimes}\label{sec:related}
A mathematical study of modes for the wave equation on Schwarzschild-adS spacetimes was initiated in~\cite{Gan14}. Note also that for a general class of equations, boundary conditions and asymptotically anti-de Sitter black holes, the papers~\cite{War15,Gan18} provided general definitions of modes and general Fredholm alternative type results.\\

Numerical results for Teukolsky equations on Kerr-anti-de Sitter black holes have been obtained starting with~\cite{Car.Lem01,Mos.Nor02} in the Schwarzschild-adS case. We mention in particular the remarkable~\cite{Car.Dia.Har.Leh.San14} where modes for the full range of admissible Kerr-adS black hole parameters are numerically computed and which in particular corroborates Theorem~\ref{thm:main} and Theorem~\ref{thm:fullmode}. In that work, the main observation is that mode stability holds as long as the Hawking-Reall bound~\eqref{est:HRboundoriginal} is satisfied, while, as soon as the Hawking-Reall bound is violated, there exists growing modes $\Im(\om)>0$.\footnote{The closer one is from the Hawking-Reall regime, the larger the azimuthal number $m$ of the growing mode must be. This is consistent with our observations of the degeneracies of the potential right when $a>kr_+^2$ provided that $|m|\to+\infty$. See Remark~\ref{rem:sharpHR}.} These growing modes come from exponentially decaying modes of the Hawking-Reall regime $\Im(\om)<0$, which have crossed the real axis $\Im(\om)=0$ exactly at stationary frequencies $\om=m\om_+$ after leaving the Hawking-Reall regime. In particular, this gives evidence for the existence of periodic stationary solutions when the Hawking-Reall bound is violated (see~\cite{Dia.San.Way15} for further works in that direction). We mention that the boundary conditions used in~\cite{Car.Dia.Har.Leh.San14} are Robin boundary conditions which were first derived in~\cite{Dia.San13}.  These Robin conditions are re-derived in Appendix~\ref{sec:Robin}, where we use them to provide an alternative (to the one of Section~\ref{sec:TSconslaws}, described in Section~\ref{sec:overview}) proof of mode stability in the non-stationary case.

\subsection{Final comments}\label{sec:finalcomments}
We comment briefly on how the above results will be used for our program (see the discussions before and at the end of Section~\ref{sec:introTeuk}). In a follow-up paper we will prove the following result:\\

  \emph{All solutions to the Teukolsky system~\eqref{Teusys} on Kerr-adS spacetimes satisfying the Hawking-Reall bound~\eqref{est:HRboundoriginal} and a slow rotation hypothesis ($ak\ll 1$ is sufficient) are bounded and decay in time with an inverse logarithmic decay rate.}\\

The proof of the theorem will exploit that for sufficiently large angular parameters $\ell \geq \ell_c(M,a,k)$, the (inhomogeneous) Regge-Wheeler equations produced by the physical space Chandrasekhar transformation theory are perturbations (under the slow rotation hypothesis) of decoupled wave equations with positive conserved energies. See the Regge-Wheeler equations~\eqref{Teusys2}, the conserved energies~\eqref{eq:energySadS} and Remark~\ref{rem:coercivity} in the Schwarzschild-adS case. The boundedness and decay results of~\cite{Hol09,Hol.Smu13} can be applied to such equations, and these properties can be integrated back to the original Teukolsky solutions (see the proof of Theorem~\ref{thm:boundedSadS}). In the finite angular regime $2\leq \ell \leq \ell_c$, where the Regge-Wheeler potential is negative, the above analysis fails. However, here Theorem \ref{thm:main} or the full mode stability result of Theorem~\ref{thm:fullmode} in a resonance expansion of the Teukolsky quantities can be used to obtain directly the boundedness and decay.\\ 

Finally, a third paper will be devoted to obtaining the full linear stability by linearising the Einstein equations in a suitable gauge and using the above theorem.\\


\emph{Mathematica notebook.} A Mathematica notebook containing the verifications of all the computations of the present article can be found here~\url{https://github.com/OGR38/Teukolsky_Kerr_adS}.\\

\emph{Acknowledgements.} Both authors thank Sam Collingbourne, Christopher Kauffman, Rita Teixeira da Costa and Claude Warnick for interesting discussions around this project. Special thanks to Allen Fang for precious insights on resonance theory and pointing out references for Theorem~\ref{thm:fullmode}. O.G.~thanks Corentin Cadiou for useful Mathematica coding suggestions. G.H.~acknowledges support by the Alexander von Humboldt Foundation in the framework of the Alexander von Humboldt Professorship endowed by the Federal Ministry of Education and Research as well as ERC Consolidator Grant 772249. Both authors acknowledge funding through Germany’s Excellence Strategy EXC 2044 390685587, Mathematics M\"unster: Dynamics-Geometry-Structure.\\

\emph{Data Availability Statement.} Data sharing not applicable to this article as no data sets were generated or analysed during the current study.\\

\emph{Conflict of interest.} The authors have no competing interests to declare that are relevant to the content of this article. Funding bodies are acknowledged in the Section ``Acknowledgements''.\\

\section{The radial and angular Teukolsky equations}\label{sec:radang}
Following~\cite{Kha83}, we have preliminary definitions of frequency-dependent radial and angular operators.
\begin{definition}
  Let $\om\in\RRR$ and $m\in\ZZZ$. Let
  \begin{align*}
    K & := am-\om(r^2+a^2), & H & := a\om\sin\varth-m\csc\varth. 
  \end{align*}
  For $n\in\ZZZ$, we define the following radial operators
  \begin{align*}
    \DD_n & := \pr_r + i\frac{\Xi K}{\De} + n \pr_r\log\De, & \DD_n^\dg & := \pr_r -i\frac{\Xi K}{\De} + n\pr_r\log\De,
  \end{align*}
  and the following angular operators\footnote{The definitions of the two operators $\LL$ and $\LL^\dag$ were wrongly interchanged in~\cite{Kha83}.}
  \begin{align*}
    \LL_n & := \pr_\varth - \frac{\Xi H}{\De_\varth} + n \pr_\varth(\log(\sqrt{\De_\varth}\sin\varth)), & \LL^\dg_n & := \pr_\varth + \frac{\Xi H}{\De_\varth} + n \pr_\varth(\log(\sqrt{\De_\varth}\sin\varth)).
  \end{align*}
\end{definition}

Let us define the following radial and angular decompositions. 
\begin{lemma}\label{lem:angulardecompo}
  Let $\om\in\RRR$. There exists real functions $S^\om_{m\ell}(\varth)$ and eigenvalues $\la^\om_{m\ell}\in\RRR$ indexed by $\ell\geq 2, |m|\leq\ell$ which are solutions to the eigenvalue equation
  \begin{align}\label{eq:defla}
    0 & = \sqrt{\De_\varth}\LL_{-1}^\dg\sqrt{\De_\varth}\LL_2S^\om_{m\ell}+\le(-6a\Xi\om\cos\varth+6k^2a^2\cos^2\varth+\la^\om_{m\ell}\ri)S^\om_{m\ell} =: \LLL^{m,\om}[\la_{m\ell}^\om] S_{m\ell}^\om,
  \end{align}
  with asymptotics (together with the consistent asymptotics for derivatives)
  \begin{align}\label{est:limitsStheta}
    \begin{aligned}
      S_{m\ell}^\om(\varth) & = O_{\varth\to0}((1-\cos\varth)^{(m+2)/2}), & S_{m\ell}^\om(\varth) & = O_{\varth\to\pi}((\cos\varth+1)^{(m-2)/2}), && \text{if $m>2$},\\
      S_{m\ell}^\om(\varth) & = O_{\varth\to0}((1-\cos\varth)^{(m+2)/2}), & S_{m\ell}^\om(\varth) & = O_{\varth\to\pi}((\cos\varth+1)^{(2-m)/2}), && \text{if $|m| \leq 2$},\\
      S_{m\ell}^\om(\varth) & = O_{\varth\to0}((1-\cos\varth)^{(-2-m)/2}), & S_{m\ell}^\om(\varth) & = O_{\varth\to\pi}((\cos\varth+1)^{(2-m)/2}), && \text{if $m<-2$}.
    \end{aligned}
  \end{align}
  The functions $S^\om_{m\ell}(\varth)e^{+im\varphi}$ are smooth spin-$+2$-weighted complex functions and the family $\le(S^\om_{m\ell}(\varth)e^{im\varphi}\ri)_{\ell\geq 2,|m|\leq\ell}$ forms an $L^2$-Hilbert basis of the smooth spin-$+2$-weighted complex functions. Moreover, the eigenfunctions $S^{\om}_{m\ell}$ and eigenvalues $\la^\om_{m\ell}$ are analytic in the admissible Kerr-adS black hole parameters $(a,k)$ and in the frequency $\om$, and in the Schwarzschild-adS case $a=0$ we have
  \begin{align*}
    \la^\om_{m\ell} & = \ell(\ell+1)-2.
  \end{align*}
\end{lemma}
\begin{proof}
  The proof is a consequence of Sturm-Liouville theory, see the proof of~\cite[Proposition 2.1]{Tei20}. The asymptotics~\eqref{est:limitsStheta} are given by the only regular branches obtained by Frobenius analysis at $\varth=0$ and $\varth=\pi$.
\end{proof}

\begin{definition}\label{def:radialdecomp}
  Let $\om\in\RRR$ and let $\al^{[\pm2]}$ be two smooth spin-$\pm2$-weighted complex modes with frequency $\om$. We define the associated \emph{radial Teukolsky quantities} to be the radial functions $\le(R_{[\pm2],\om}^{m\ell}(r)\ri)_{\ell\geq 2,|m|\leq \ell}$ projections of the renormalised functions $\widetilde{\al}^{[\pm2]}$ on the Hilbert basis of Lemma~\ref{lem:angulardecompo}. That is,
  \begin{align*}
    \widetilde{\al}^{[\pm2]}(t,r,\varth,\varphi) & = \sum_{\ell\geq 2}\sum_{|m|\leq \ell} e^{\mp i\om t} e^{\pm im\varphi} R_{[\pm2],\om}^{m\ell}(r)  S^\om_{m\ell}(\varth),
  \end{align*}
  for all $(t,r,\varth,\varphi)\in\RRR\times(r_+,+\infty)\times(0,\pi)\times(0,2\pi)$.
\end{definition}
\begin{remark}
  The complex conjugate of a smooth spin-$\pm2$-weighted complex function is a smooth spin-$\mp2$-weighted complex function. The decomposition of Definition~\ref{def:radialdecomp} in the spin-$-2$ case is obtained using this conjugacy.
\end{remark}

The Teukolsky equations satisfy the following separability properties.
\begin{lemma}[Radial Teukolsky equations]\label{lem:radTeuk}
  Let $\om\in\RRR$. Let $\al^{[\pm2]}$ be two spin-$\pm2$-weighted complex modes with frequency $\om$. Then, $\al^{[\pm2]}$ satisfy the Teukolsky equations~\eqref{eq:Teukoriginal} if and only if the associated radial Teukolsky quantities $R^{m\ell}_{[\pm2],\om}$ as defined in Definition~\ref{def:radialdecomp} satisfy the following \underline{radial Teukolsky equations}
  \begin{align}
    \label{eq:radialTeuk}
    \begin{aligned}
      0 & = \widetilde{\mathfrak{I}}^{m,\om}_{[+2]}[\la_{m\ell}^\om] R_{[+2],\om}^{m\ell}  := \frac{\De^3}{(r^2+a^2)^{7/2}}\mathfrak{I}^{m,\om}_{[+2]}[\la_{m\ell}^\om]\le(\frac{(r^2+a^2)^{3/2}}{\De^2}R_{[+2],\om}^{m\ell}\ri),\\
      0 & = \widetilde{\mathfrak{I}}^{-m,-\om}_{[-2]}[\la_{m\ell}^\om] R_{[-2],\om}^{m,\ell}  := \frac{\De}{(r^2+a^2)^{7/2}}\mathfrak{I}^{-m,-\om}_{[-2]}[\la_{m\ell}^\om]\le((r^2+a^2)^{3/2}R_{[-2],\om}^{m,\ell}\ri),
    \end{aligned}
  \end{align}
  where
  \begin{align*}
    \mathfrak{I}^{m,\om}_{[+2]}[\la_{m\ell}^\om]  & := \De\DD_1\DD_2^\dg + 6i\Xi\om r+6k^2r^2-\la^\om_{m\ell},\\
    \mathfrak{I}^{-m,-\om}_{[-2]}[\la_{m\ell}^\om] & := \De\DD_{-1}\DD_0^{\dg}+6i\Xi\om r+6k^2r^2-\la^\om_{m\ell}.
  \end{align*}
  Moreover,
  \begin{itemize}
  \item the functions $\widetilde{\al}^{[\pm2]}$ satisfy the conformal anti-de Sitter boundary conditions at infinity~\eqref{eq:defbdycond} if and only if
    \begin{align}\label{eq:radialbdyconditions}
      \begin{aligned}
        R_{[+2],\om}^{m\ell}-\le(R_{[-2],\om}^{m\ell}\ri)^\ast & \to 0, & r^2\pr_rR_{[+2],\om}^{m\ell}+r^2\pr_r\le(R_{[-2],\om}^{m\ell}\ri)^\ast & \to 0, 
      \end{aligned}
    \end{align}
    when $r\to +\infty$,
  \item if $\om-m\om_+\neq 0$, all solutions to the two radial Teukolsky equations~\eqref{eq:radialTeuk} can be written respectively as
    \begin{align}\label{eq:radialregconditionbasisdecomp}
      \begin{aligned}
      R_{[+2],\om}^{m\ell} & = A_{[+2],\om,\HH^+}^{m\ell}R_{[+2],\om,\HH^+}^{m\ell} + A_{[+2],\om,\HH^-}^{m\ell}R_{[+2],\om,\HH^-}^{m\ell},\\
      R_{[-2],\om}^{m\ell} & = A_{[-2],\om,\HH^+}^{m\ell}R_{[-2],\om,\HH^+}^{m\ell} + A_{[-2],\om,\HH^-}^{m\ell}R_{[-2],\om,\HH^-}^{m\ell},
      \end{aligned}
    \end{align}
    where $A_{[\pm2],\om,\HH^{\pm}}^{m\ell}$ are complex-valued constants and where $R_{[+2],\om,\HH^+}^{m\ell},R_{[+2],\om,\HH^-}^{m\ell}$ and $R_{[-2],\om,\HH^+}^{m\ell},R_{[-2],\om,\HH^-}^{m\ell}$ are respectively the unique basis of solutions to the two Teukolsky equations~\eqref{eq:radialTeuk} satisfying\footnote{From Lemma~\ref{lem:Teukeqrel}, it will follow that $R_{[+2],\om,\HH^+}^{m\ell} = R_{[-2],\om,\HH^{-}}^{m\ell}$ and $R_{[+2],\om,\HH^-}^{m\ell} = R_{[-2],\om,\HH^+}^{m\ell}$.}
      \begin{align}\label{eq:radialregconditionsbasis}
        \begin{aligned}
        R_{[+2],\om,\HH^+}^{m\ell} & = e^{-i\Xi(\om-m\om_+) r^\ast}F_{[+2],\om,\HH^+}^{m\ell}(r), & R_{[+2],\om,\HH^-}^{m\ell} & = e^{+i\Xi(\om-m\om_+)r^\ast} \De^2F_{[+2],\om,\HH^-}^{m\ell}(r),\\
        R_{[-2],\om,\HH^+}^{m\ell} & = e^{+i\Xi(\om-m\om_+) r^\ast}\De^2F_{[-2],\om,\HH^+}^{m\ell}(r), & R_{[-2],\om,\HH^-}^{m\ell} & = e^{-i\Xi(\om-m\om_+)r^\ast}F_{[-2],\om,\HH^-}^{m\ell}(r),
        \end{aligned}
      \end{align}
      where $F_{[\pm2],\om,\HH^{\pm}}^{m\ell}:[r_+,+\infty)\to\CCC$ are smooth functions of $r$ satisfying $F_{[\pm2],\om,\HH^{\pm}}^{m\ell}(r_+)=1$. With these notations, the regularity conditions at the horizon~\eqref{eq:defreghor} for the functions $\widetilde{\al}^{[\pm2]}$ is equivalent to
      \begin{subequations}\label{eq:radialreggeneral}
        \begin{align}\label{eq:radialregconditions}
          \begin{aligned}
            A_{[+2],\om,\HH^-}^{m\ell}=A_{[-2],\om,\HH^-}^{m\ell}=0,
          \end{aligned}
        \end{align}
      \item if $\om-m\om_+=0$, the regularity conditions~\eqref{eq:defreghor} at the horizon for the functions $\widetilde{\al}^{[\pm2]}$ imply
    \begin{align}\label{eq:radialregstat}
      R_{[\pm2],\om}^{m\ell} & = \De^2 F_{[\pm2],\om}^{m\ell}(r),
    \end{align}
    \end{subequations}
    with $F_{[\pm2],\om}^{m\ell}:[r_+,+\infty)\to\CCC$ smooth functions of $r$.
  \end{itemize}
\end{lemma}
\begin{remark}
  The proof of the last part~\eqref{eq:radialregstat} will make use of some of the algebraic identities derived (independently) only in Section~\ref{sec:TStrafos}. We have included the statement here to collect all relevant properties in one lemma. 
\end{remark}
\begin{proof}
  Equations~\eqref{eq:radialTeuk} follow from~\cite[Equations (3.19) and (3.21)]{Kha83}.\footnote{We think that the definition of $\Phi_4$ in~\cite{Kha83} should be $\Phi_4:=(r-ia\cos\varth)^4\psi_4$ instead of $\Phi_4:=(r-ia\cos\varth)^3\psi_4$.} See also~\cite[Equation (3.9)]{Cas.Tei22}. The equivalence of the boundary conditions~\eqref{eq:defbdycond} and~\eqref{eq:radialbdyconditions} is directly obtained from the angular and radial decomposition of Definition~\ref{def:radialdecomp}. The existence of a unique basis of solutions satisfying~\eqref{eq:radialregconditionsbasis} is obtained by a direct Frobenius expansion at the regular singular point $r=r_+$. To prove~\eqref{eq:radialregconditions} and~\eqref{eq:radialregstat} we first define 
  \begin{align*}
    c^2(M,a,k) & := \frac{1}{r_+^2+a^2}\pr_r\De(r=r_+)>0.
  \end{align*}
  From the regularity condition at the horizon~\eqref{eq:defreghor} for $\widetilde{\al}^{[+2]}$ and the angular decomposition of Definition~\ref{def:radialdecomp} we have
  \begin{align}\label{eq:Lbradial}
    \begin{aligned}
      (-i\Xi(\om-m\om_+)-\pr_{r^\ast})R_{[+2],\om}^{m\ell} & = \Lb\le(e^{-i\om t}e^{im\varphi}R_{[+2],\om}^{m\ell}\ri) + O_{r\to r_+}(\De)\\
      & =  O_{r\to r_+}(\De) \\
      & = O_{r^\ast\to-\infty}(e^{c^2(a,M,k) r^\ast}).
      \end{aligned}
  \end{align}
  Integrating~\eqref{eq:Lbradial} in $r^\ast$, we obtain that
  \begin{align*}
    R_{[+2],\om}^{m\ell}e^{+i\Xi(\om-m\om_+) r^\ast} & = O(1),
  \end{align*}
  when $r^\ast\to-\infty$. From an analogous limit for $\De^{-2}R_{[-2],\om}^{m\ell}$ as well as higher regularity asymptotics follow similarly and we infer~\eqref{eq:radialregconditions}. In the case $\om-m\om_+=0$, we deduce from~\eqref{eq:Lbradial} that
  \begin{align*}
    \pr_rR_{[+2],\om}^{m\ell} & = O_{r\to r_+}(1),
  \end{align*}
  and by using the regularity conditions~\eqref{eq:defreghor} further
  \begin{align}\label{eq:prnR+2}
    \pr_r^nR_{[+2],\om}^{m\ell} & = O_{r\to r_+}(1),
  \end{align}
  for all $n\in\mathbb{N}$. A Frobenius analysis at $r=r_+$ for the radial Teukolsky equation~\eqref{eq:radialTeuk} shows that there are two non-trivial solutions $R_{\mathrm{irreg}}$ and $R_{\mathrm{reg}}$ with
  \begin{align*}
    R_{\mathrm{irreg}}(r) & = f(r) + C\De^2 \log(\De) g(r) , & R_{\mathrm{reg}} & = \De^2 g(r),
  \end{align*}
  where $f,g$ are smooth functions on $[r_+,+\infty)$ such that $f(r_+)=g(r_+)=1$, and where $C\in\mathbb{C}$. Anticipating on the Teukolsky-Starobinsky identity~\eqref{eq:invTS}, we have
  \begin{align*}
    \wp(m,m\om_+,\la^{m\om_+}_{m\ell}) R_{\mathrm{irreg}}(r) & = \frac{\De^2}{(r^2+a^2)^{3/2}}\le(\DD_0\ri)^4\De^2(\DD_0^\dg)^4\le((r^2+a^2)^{3/2}\le(f(r) + C\De^2 \log(\De) g(r)\ri)\ri), 
  \end{align*}
  for all $r>r_+$, where $\wp(m,m\om_+,\la^{m\om_+}_{m\ell})$ is the radial Teukolsky-Starobinsky constant defined in Lemma~\ref{lem:radTS}. Hence, since $f$ is smooth and does not vanish at $r=r_+$, and since $\DD_0$, $\DD_0^\dg$ are regular operators at $r=r_+$ in the stationary case $\om=m\om_+$, one has
  \begin{align*}
    C = 0 \implies \wp(m,m\om_+,\la^{m\om_+}_{m\ell})=0.
  \end{align*}
  The radial Teukolsky-Starobinsky constant $\wp(m,m\om_+,\la^{m\om_+}_{m\ell})$ being always striclty positive (see Lemmas~\ref{lem:angTS} and \ref{lem:radTS}), we infer an immediate contradiction and hence that $C\neq0$. Hence, we deduce from~\eqref{eq:prnR+2} that $R_{[+2],\om}^{m\ell}$ must be proportional to $R_{\mathrm{reg}}$ (since $R_{\mathrm{irreg}}$ is not smooth at $r=r_+$ because $C\neq0$). Along the same lines, we also have that $R_{[-2],\om}^{m\ell}$ is proportional to $R_{\mathrm{reg}}$ and this finishes the proof of~\eqref{eq:radialregstat}.
\end{proof}



A direct computation gives the following relations for the radial Teukolsky operators.
\begin{lemma}\label{lem:Teukeqrel}
  We have
  \begin{align*}
    \widetilde{\mathfrak{I}}_{[+2]}^{m,\om}[\la_{m\ell}^\om] & = \widetilde{\mathfrak{I}}_{[-2]}^{-m,-\om}[\la_{m\ell}^\om] = \le(\widetilde{\mathfrak{I}}_{[-2]}^{m,\om}[\la_{m\ell}^\om]\ri)^\ast = \le(\widetilde{\mathfrak{I}}_{[+2]}^{-m,-\om}[\la_{m\ell}^\om]\ri)^\ast \\
                                                             & = \pr_{r^\ast}^2 +\le(\pr_{r^\ast}\le(\log\le(\frac{(r^2+a^2)^4}{\De^2}\ri)\ri)\ri)\pr_{r^\ast} + V^{m,\om}[\la^\om_{m\ell}],
  \end{align*}
  with
  \begin{align*}
    V^{m,\om}[\la^\om_{m\ell}] & := \frac{\De}{(r^2+a^2)^2}\le(8i\Xi\om r+6k^2r^2-\la^\om_{m\ell}\ri)+\frac{\pr_r\De}{(r^2+a^2)^2}2i\Xi K  + \frac{\Xi^2K^2}{(r^2+a^2)^2} \\
                               & \quad -3r\frac{\De\pr_r\De}{(r^2+a^2)^{3}} + 3(2r^2+a^2)\frac{\De^2}{(r^2+a^2)^4}.
  \end{align*}
\end{lemma}

\section{Teukolsky-Starobinsky transformations} \label{sec:TStrafos}
\begin{lemma}[Angular Teukolsky-Starobinsky transformations]\label{lem:angTS}
  Let $S_{m\ell}^\om$ be a solution to~\eqref{eq:defla} with $\la=\la^\om_{m\ell}$ such that $S_{m\ell}^\om e^{im\varphi}$ is a smooth spin-$+2$-weighted function. Let us define the \emph{associated Teukolsky-Starobinsky quantity} $S_{m\ell,c}^\om$ by
  \begin{align*}
    S_{m\ell,c}^\om & := \sqrt{\De_\varth}\LL_{-1}\sqrt{\De_\varth}\LL_{0}\sqrt{\De_\varth}\LL_1\sqrt{\De_\varth}\LL_2 S_{m\ell}^\om.
  \end{align*}
  Then, we have
  \begin{align*}
    0 & = \LLL^{-m,-\om}[\la^\om_{m\ell}]S_{m\ell,c}^\om = \sqrt{\De_\varth}\LL_{-1}\sqrt{\De_\varth}\LL_2^\dg S^\om_{m\ell,c}+\le(6a\Xi\om\cos\varth+6k^2a^2\cos^2\varth+\la^\om_{m\ell}\ri)S^\om_{m\ell,c}.
  \end{align*}
  Moreover, we have the following formula
  \begin{align*}
    \sqrt{\De_\varth}\LL_{-1}^\dg\sqrt{\De_\varth}\LL_{0}^\dg\sqrt{\De_\varth}\LL_1^\dg\sqrt{\De_\varth}\LL_2^\dg S_{m\ell,c}^\om & = \aleph(m,\ell,\la^{\om}_{m\ell}) S^{\om}_{m\ell},
  \end{align*}
  where $\aleph(m,\om,\la^\om_{m\ell})$ is the \emph{angular Teukolsky-Starobinsky constant} defined by
  \begin{align*}
    \aleph(m,\om,\la^\om_{m\ell}) & := 4 (\la_{m\ell}^\om)^2 + 4 (\la_{m\ell}^\om)^3 + (\la_{m\ell}^\om)^4 + 8 a \la_{m\ell}^\om (6 + 5 \la_{m\ell}^\om) (\Xi m)(\Xi \om)  \\
                                  & \quad - 144 a^3 (\Xi m)(\Xi\om) \le(k^2 (-2 + \la_{m\ell}^\om) + 2 (\Xi\om)^2\ri) + 4 a^4 \le(k^2 (-6 + \la_{m\ell}^\om) + 6 (\Xi\om)^2\ri)^2 \\
                                  & \quad + 4 a^2 k^2 \la_{m\ell}^\om \le(-12 - 4 \la_{m\ell}^\om + (\la_{m\ell}^\om)^2 + 24 (\Xi m)^2\ri) + 8a^2 \le(6 \la_{m\ell}^\om - 5 (\la_{m\ell}^\om)^2 + 18 (\Xi m)^2\ri) (\Xi\om)^2 \\
                                  & > 0.
  \end{align*}
\end{lemma}
\begin{proof}
  All these identities except the positivity of the angular Teukolsky-Starobinsky constant are direct computations and are left to the reader.\footnote{\label{foo:Mathematica}A Mathematica notebook containing the verifications of all the computations of the present article is included as an ancillary arXiv file. It can also be found here~\url{https://github.com/OGR38/Teukolsky_Kerr_adS}.} To obtain the positivity of $\aleph$ we revisit arguments from the proof of~\cite[Lemmas 2.9 and 2.11]{Cas.Tei21}.\footnote{We think that there is small gap in the proof of~\cite[Lemma 2.11]{Cas.Tei21} in the $|m|\leq 1$ case which can be fixed by setting $k=0$ in our argument.}\;
  Note that the $a=0$ case is immediate from the expression of $\aleph$ and the expression of $\la_{m\ell}^\om$ (see Lemma~\ref{lem:angulardecompo}) and we shall therefore assume that $a>0$. Note also that the problem is invariant under the transformation
  \begin{align*}
    \varth & \to \pi -\varth, & m & \to -m, & \om & \to -\om,
  \end{align*}
  and we shall therefore assume $m\geq0$. By integration by parts we have
  \begin{align}\label{eq:proofIBPangTS}
    \begin{aligned}
      & \aleph\int_{0}^\pi |S_{m\ell}^\om|^2\,\sin\varth\d\varth \\
      = & \; \int_{0}^\pi \sqrt{\De_\varth}\LL_{-1}^\dg\sqrt{\De_\varth}\LL_0^\dg\sqrt{\De_\varth}\LL_{1}^\dg\sqrt{\De_\varth}\LL_2^\dg \sqrt{\De_\varth}\LL_{-1}\sqrt{\De_\varth}\LL_0\sqrt{\De_\varth}\LL_1\sqrt{\De_\varth}\LL_2S_{m\ell}^\om S_{m\ell}^\om \,\sin\varth\d\varth \\
    = & \; \int_{0}^\pi \le|\sqrt{\De_\varth}\LL_{-1}\sqrt{\De_\varth}\LL_0\sqrt{\De_\varth}\LL_1\sqrt{\De_\varth}\LL_2S_{m\ell}^\om\ri|^2\,\sin\varth\d\varth \geq 0,
    \end{aligned}
  \end{align}
  from which we obtain $\aleph\geq 0$. We also infer from~\eqref{eq:proofIBPangTS} that $\aleph(m,\om,\la^{\om}_{m\ell})=0$ iff
  \begin{align}\label{eq:zeroTSang}
     S^\om_{m\ell,c} = \sqrt{\De_\varth}\LL_{-1}\sqrt{\De_\varth}\LL_0\sqrt{\De_\varth}\LL_1\sqrt{\De_\varth}\LL_2S_{m\ell}^\om=0.
  \end{align}

  In the $m \geq 2$ case, a direct computation using the asymptotics~\eqref{est:limitsStheta} gives
  \begin{align*}
    \begin{aligned}
      (1-\cos\varth)^{-(m-2)/2}S^\om_{m\ell,c}(\varth) & \xrightarrow{\varth\to0} -4\Xi^2(m-1)m(m+1)(m+2)\le(\lim_{\varth\to0}(1-\cos\varth)^{-(m+2)/2}S_{m\ell}^\om(\varth)\ri),
    \end{aligned}
  \end{align*}
  Thus,~\eqref{eq:zeroTSang} holds only if $\lim_{\varth\to0}(1-\cos\varth)^{-(m+2)/2}S_{m\ell}^\om(\varth) = 0$. 
  From a Frobenius analysis of the ODE~\eqref{eq:defla} at $\varth=0$ this implies that $S_{m\ell}^\om = 0$.\\

  In the $m=0$ case, a computation gives
  \begin{align*}
    \begin{aligned}
      (1-\cos\varth)^{-1}S^\om_{0\ell,c}(\varth) & \xrightarrow{\varth\to0} \le(\ka_0 + \ka_0'\ri)\le(\lim_{\varth\to0}(1-\cos\varth)^{-1}S_{0\ell}^\om(\varth)\ri),\\
      (1-\cos\varth)^{-1}S^\om_{0\ell,c}(\varth) & \xrightarrow{\varth\to\pi} (\ka_0-\ka_0') \le(\lim_{\varth\to\pi}(1-\cos\varth)^{-1}S_{0\ell}^\om(\varth)\ri).
    \end{aligned}
  \end{align*}
  with
  \begin{align*}
    \ka_0 & := \la_{0\ell}^\om (2 + \la_{0\ell}^\om)  - 24 a^4 k^2 \om^2 + 12 a^6 k^4 \om^2 + 2 a^2 \le(k^2 (-6 + \la_{0\ell}^\om) + 6 \om^2\ri), & \ka_0' & :=  8 a \la_{0\ell}^\om \Xi\om.
  \end{align*}
  From a Frobenius analysis of~\eqref{eq:defla} at $\varth=0,\pi$, identity~\eqref{eq:zeroTSang} can hold for non-trivial $S_{0\ell,c}^\om$ only if $\ka_0=\ka'_0=0$. Using that $a>0$, this only holds in the following four cases
  \begin{subequations}\label{eq:m0vanishingTS}
    \begin{align}
      m & = 0, & \la & = 0, & \Xi\om & = \pm k,\label{eq:m0la0}
                                       \intertext{or}
      m & = 0, & \la & = -1-a^2k^2 \pm \sqrt{1+14a^2k^2+a^4k^4}, & \Xi\om & = 0.\label{eq:m0om0}
    \end{align}
    \end{subequations}
    In the $m=1$ case, a computation gives
    \begin{align*}
      \begin{aligned}
        (1-\cos\varth)^{-1/2}S_{1\ell,c}^\om(\varth) & \xrightarrow{\varth\to0} -\ka_1 \le(\lim_{\varth\to0}(1-\cos\varth)^{-3/2}S_{1\ell}^\om(\varth)\ri),
      \end{aligned}
    \end{align*}
    with
    \begin{align*}
      \ka_1 & := 12 \Xi (6 a^2 k^2 + \la + 6 a \Xi \om).
    \end{align*}
    Thus we shall assume that $\ka_1=0$ (\emph{i.e.} $\la = - 6a^2k^2-6a\Xi\om$) otherwise $S_{1\ell}^\om=0$. Under this assumption, a direct computation gives that
    \begin{align}\label{eq:ODETSm1}
      \begin{aligned}
        S_{1\ell,c}^\om & = -8\ka_1'a^3\sin^2\varth\Deth^{-2}\le(f(\varth)S_{1\ell}^\om(\varth) - \sin\varth\Deth\pr_\varth S_{1\ell}^\om(\varth)\ri), 
      \end{aligned}
    \end{align}
    with
    \begin{align*}
      \ka_1' & := (a k^2 + \Xi\om) (k - a k^2 - \Xi\om) (k + a k^2 + \Xi\om)\\
      f(\varth) & := 2+\cos\varth+a\Xi\om\sin^2\varth +a^2k^2(1-2\cos\varth-3\cos^2\varth+\cos^3\varth).
    \end{align*}
    Assume that $\ka'_1\neq0$. Then, integrating the first order ODE given by combining~\eqref{eq:zeroTSang} and~\eqref{eq:ODETSm1}, we have
    \begin{align*}
      S_{1\ell}^\om(\varth) & = C (1-\cos\varth)^{-3/2}(1+\cos\varth)^{1/2}\Deth^2\exp\le(-\le(ak+\frac{\Xi\om}{k}\ri)\mathrm{arctanh}(ak\cos\varth)\ri),
    \end{align*}
    where $C$ is the integration constant. The above does not satisfy the required asymptotic $\sim(1-\cos\varth)^{3/2}$ when $\varth\to 0$ except if $C=0$ and $S_{1\ell}^\om=0$. Thus, we must have $\ka'_1=0$. The vanishing of $\ka_1$ and $\ka'_1$ is equivalent to the following three cases
    \begin{subequations}\label{eq:m1vanishingTS}
      \begin{align}
        m & = 1, & \la & = 0, & \Xi\om & = -a k^2,\label{eq:m1la0}
                                         \intertext{or}
        m & = 1, & \la & = \mp 6ak, & \Xi\om & = \pm k(1 \mp ak).\label{eq:m1lapos}
      \end{align}
    \end{subequations}
    We now prove that there are no solutions to equation~\eqref{eq:defla} which are regular at $\varth=0$ and $\varth=\pi$ for triplets $(m,\la,\Xi\om)$ as given in the seven cases~\eqref{eq:m0la0}~\eqref{eq:m0om0}~\eqref{eq:m1la0}~\eqref{eq:m1lapos}.\footnote{For these values $(m,\la,\Xi\om)$, a direct computation gives that for any function $S$ which is a solution of $\LL^{m,\om}[\la]S=0$, its Teukolsky-Starobinsky transformation $S_c = \sqrt{\De_\varth}\LL_{-1}\sqrt{\De_\varth}\LL_0\sqrt{\De_\varth}\LL_1\sqrt{\De_\varth}\LL_2 S$ is \emph{identically zero} (\emph{i.e.} identity~\eqref{eq:zeroTSang} is satisfied).} Using that $S_{m\ell}^\om$ is regular at $\varth=0,\pi$, multiplying~\eqref{eq:defla} with $S=S_{m\ell}^\om$ and integrating by parts gives
    \begin{align}\label{eq:energyidentityTS}
      0 & = \int_{0}^\pi\le(\De_\varth|\pr_\varth S|^2 + G(\varth) |S|^2\ri)\sin\varth\d\varth,
    \end{align}
    with
    \begin{align*}
      G(\varth) & := \quar\De_\varth^{-1}\big(-8 + 34 X^2 - 17 X^4 \pm 16 X (1 - X^2) \cos(\varth) - 10 X^2 \cos(2\varth) \\
                & \quad + X^4 \cos(4\varth) + 16 (1 - X^2)^2 \csc(\varth)^2\big) && \text{in the~\eqref{eq:m0la0} cases},\\
      G(\varth) & := 3 + X + X^2 \mp \sqrt{1 + 14 X^2 + X^4} - X^2 \cos(2 \varth) -4 \De_\varth^{-1} (1 - X^2) + 4 (1 - X^2) \csc(\varth)^2 && \text{in the~\eqref{eq:m0om0} cases},\\
      G(\varth) & := \quar \big(4 (3 + X^2) + \frac{2 - 2 X^2}{1 + \cos\varth} +4(1-X^2)\De_\varth^{-1}\le(-5+2X-X^2+4(-1+X)X\cos\varth\ri)\\
                & \quad - 4 X^2 \cos(2 \varth) + 9 (1 - X^2) \csc(\varth/2)^2\big) && \text{in the~\eqref{eq:m1la0} case},\\
      G(\varth) & := \quar \big(12 + 8 X (\pm 2 + X) + \frac{2 - 2 X^2}{1 + \cos\varth} + 4\De_\varth^{-1} (1 - X^2) (-5 + 4 X \cos\varth) \\
                & \quad - 4 X^2 \cos(2\varth) + 9 (1 - X^2) \csc(\varth/2)^2\big) && \text{in the~\eqref{eq:m1lapos} cases},
    \end{align*}
    where $X:=ak$. The functions $G$ are polynomials in $X$ and $x:=\cos\varth$ which can be checked to be positive for $-1<x<1$. Thus, we deduce from~\eqref{eq:energyidentityTS} that $S=S_{m\ell}^\om=0$ if $m=0$ or $m=1$ and this finishes the proof of the lemma.
\end{proof}

\begin{lemma}[Radial Teukolsky-Starobinsky transformations]\label{lem:radTS}
  Let $R_{[\pm2],\om}^{m\ell}$ be solutions to the radial Teukolsky equation~\eqref{eq:radialTeuk}. Let us define the \emph{associated Teukolsky-Starobinsky quantities} $R_{[\pm2],\om,c}^{m\ell}$ by
  \begin{align}\label{eq:defTeukStar}
    \begin{aligned}
      R_{[+2],\om,c}^{m\ell} & := \frac{\De^2}{(r^2+a^2)^{3/2}}(\DD_0^{\dg})^4\le((r^2+a^2)^{3/2}R_{[-2],\om}^{m\ell}\ri), \\
      R_{[-2],\om,c}^{m\ell} & := \frac{\De^2}{(r^2+a^2)^{3/2}}(\DD_0^\dg)^4\le((r^2+a^2)^{3/2}R_{[+2],\om}^{m\ell}\ri). 
    \end{aligned}
  \end{align}
  Then, the Teukolsky-Starobinsky quantities satisfy the following Teukolsky equations
  \begin{align}\label{eq:TeukTS}
    \widetilde{\mathfrak{I}}_{[+2]}^{-m,-\om}[\la^\om_{m\ell}] R_{[+2],\om,c}^{m\ell} & = 0, & \widetilde{\mathfrak{I}}_{[-2]}^{m,\om}[\la^\om_{m\ell}] R_{[-2],\om,c}^{m\ell} & = 0,
  \end{align}
  and we have the following inversion formulas
  \begin{align}\label{eq:invTS}
    \begin{aligned}
       R_{[+2],\om}^{m\ell} & = \wp(m,\om,\la^\om_{m\ell})^{-1}\frac{\De^2}{(r^2+a^2)^{3/2}}\le(\DD_0\ri)^4\le((r^2+a^2)^{3/2}R_{[-2],\om,c}^{m\ell}\ri),\\
       R_{[-2],\om}^{m\ell} & = \wp(m,\om,\la^\om_{m\ell})^{-1}\frac{\De^2}{(r^2+a^2)^{3/2}}\le(\DD_0\ri)^4\le((r^2+a^2)^{3/2}R_{[+2],\om,c}^{m\ell}\ri),
    \end{aligned}
  \end{align}
  where $\wp(m,\om,\la^\om_{m\ell})$ is the \emph{radial Teukolsky-Starobinsky constant} defined by
  \begin{align*}
    \wp(m,\om,\la^\om_{m\ell}) & := 144 M^2 (\Xi\om)^2 + \aleph(m,\om,\la^\om_{m\ell}).
  \end{align*}
\end{lemma}
\begin{proof}
  The proof of the lemma are direct computations which are left to the reader.\footnoteref{foo:Mathematica}\,\footnote{As pointed out using Chamber-Moss coordinates in~\cite{Dia.San13} (see also~\cite{Cha.Mos94}), the angular and radial Teukolsky equations share a similar structure and computations carry over from one equation to the other.}
\end{proof}

We have the following lemma providing asymptotics at the horizon for Teukolsky-Starobinksy transformations.
\begin{lemma}[Teukolsky-Starobinsky horizon asymptotics]\label{lem:TShorasympt}
  Let $R_{[\pm2],\om}^{m\ell}$ be solutions to the radial Teukolsky equation~\eqref{eq:radialTeuk}. 
  The Teukolsky-Starobinsky transformations $R^{m\ell}_{[+2],\om,c}$ have the following horizon decompositions.
  \begin{itemize}
  \item If $\om-m\om_+\neq0$, we have
    \begin{align}\label{eq:radialregconditionsTSgeneral}
      \begin{aligned}
        R_{[+2],\om,c}^{m\ell} & = C(m,\om) A_{[-2],\om,\HH^+}^{m\ell} R_{[+2],-\om,\HH^+}^{-m\ell} + \wp(m,\om,\la^\om_{m\ell})\le(C^\ast(m,\om)\ri)^{-1} A_{[-2],\om,\HH^-}^{m\ell} R_{[+2],-\om,\HH^-}^{-m\ell}, \\
        R_{[-2],\om,c}^{m\ell} & = \wp(m,\om,\la^\om_{m\ell})\le(C^\ast(m,\om)\ri)^{-1} A_{[+2],\om,\HH^+}^{m\ell} R_{[-2],-\om,\HH^+}^{-m\ell} + C(m,\om) A_{[+2],\om,\HH^-}^{m\ell} R_{[-2],-\om,\HH^-}^{-m\ell},
    \end{aligned}
    \end{align}
    where $A_{[\pm2],\om,\HH^\pm}^{m\ell}$ are the same constants as in~\eqref{eq:radialregconditionbasisdecomp}, and where
    \begin{align*}
      C(m,\om) & := -\le(2i\Xi(\om-m\om_+)(r_+^2+a^2)\ri) \le(2\pr_r\De(r_+)+2i\Xi(\om-m\om_+)(r_+^2+a^2)\ri)\\
               & \quad \times \le((\pr_r\De(r_+))^2+\le(2\Xi(\om-m\om_+)(r_+^2+a^2)\ri)^2\ri), \\
               & =: -i\xi\le(\Xi(\om-m\om_+)\ri)\widetilde{C}(m,\om), 
    \end{align*}
    with $\xi:=\pr_r\De(r_+)+i\Xi(\om-m\om_+)(r_+^2+a^2)$. 
  \item If $\om-m\om_+=0$ and if $R_{[\pm2],\om}^{m\ell}$ satisfy the regularity conditions~\eqref{eq:radialregstat}, we have
    \begin{align}
      \label{eq:radialregstatTS}
      R_{[\pm2],\om,c}^{m\ell} & = \De^2 F_{[\pm2],\om,c}^{m\ell}(r),
    \end{align}
    where $F_{[\pm2],\om,c}^{m\ell}:[r_+,+\infty)\to \CCC$ are smooth functions of $r$.
  \end{itemize}
\end{lemma}
\begin{proof}
  We have 
  \begin{align*}
    \frac{\De^2}{(r^2+a^2)^{3/2}}(\DD_0^\dg)^4\le((r^2+a^2)^{3/2}R_{[+2],\om,\HH^-}^{m\ell}\ri) & = \De^2 (\DD_0^\dg)^4\le(\De^2e^{+i\Xi(\om-m\om_+)r^\ast}F(r)\ri) \\
     & = \De^2 e^{-i\Xi(\om-m\om_+)r^\ast} \pr_r^4\le(\De^2 e^{+2i\Xi(\om-m\om_+)r^\ast}\ri)F(r), 
  \end{align*}
  where, in each of the above lines, $F(r):[r_+,\infty)\to\CCC$ is always a smooth function with $F(r_+)=1$, but which in general differs from line to line. Computing further, we have
  \begin{align*}
    \pr_r^4\le(\De^2 e^{+2i\Xi(\om-m\om_+)r^\ast}\ri)F(r) & = \le(2\pr_r\De+2i\Xi(\om-m\om_+)(r_+^2+a^2)\ri)\pr_r^3\le(\De e^{+2i\Xi(\om-m\om_+)r^\ast}\ri)F(r)\\
                                                       & =  \le(2\pr_r\De+2i\Xi(\om-m\om_+)(r_+^2+a^2)\ri)\le(\pr_r\De+2i\Xi(\om-m\om_+)(r_+^2+a^2)\ri)\\
    & \quad \quad \quad \quad \times \pr_r^2\le(e^{+2i\Xi(\om-m\om_+)r^\ast}\ri)F(r)\\
    & = \le(2\pr_r\De+2i\Xi(\om-m\om_+)(r_+^2+a^2)\ri)\le(\pr_r\De+2i\Xi(\om-m\om_+)(r_+^2+a^2)\ri) \\
    & \quad \quad \quad \quad \times \le(2i\Xi(\om-m\om_+)(r_+^2+a^2)\ri) \pr_r\le(\De^{-1}e^{+2i\Xi(\om-m\om_+)r^\ast}\ri)F(r) \\
    & = - \le(2\pr_r\De+2i\Xi(\om-m\om_+)(r_+^2+a^2)\ri)\le((\pr_r\De)^2+(2\Xi(\om-m\om_+)(r_+^2+a^2))^2\ri)\\
    & \quad \quad \quad \quad \times \le(2i\Xi(\om-m\om_+)(r_+^2+a^2)\ri) \De^{-2}e^{+2i\Xi(\om-m\om_+)r^\ast}F(r).
  \end{align*}
  Using the above computations, the definition of the Teukolsky-Starobinsky transformations, the Teukolsky equation for the Teukolsky-Starobinsky transformations~\eqref{eq:TeukTS}, and the definition of the horizon basis~\eqref{eq:radialregconditionsbasis}, we deduce that
  \begin{align*}
    R_{[-2],\om,\HH^-,c}^{m\ell} = C(m,\om)e^{+i\Xi(\om-m\om_+)r^\ast}F(r) = C(m,\om)R_{[-2],-\om,\HH^-}^{-m\ell}.
  \end{align*}
  In the $R_{[-2],\om,\HH^+,c}$ case the asymptotics are obtained from the $R_{[-2],\om,\HH^-,c} = R_{[+2],\om,\HH^+,c}$ case and the inversion formula~\eqref{eq:invTS}. The regularity~\eqref{eq:radialregstatTS} are direct consequences of the smoothness of $\De^{-2}R_{[\pm2]}$ from~\eqref{eq:radialregstat}. This finishes the proof of the lemma.
\end{proof}

We have the following limits at infinity for the Teukolsky-Starobinsky transformations. The proof of the following lemma is a direct calculation which is left to the reader.\footnoteref{foo:Mathematica}
\begin{lemma}[Teukolsky-Starobinsky transmission coefficients]\label{lem:TStransmissioninfinity}
  Let $m\in\ZZZ$, $\om\in\RRR$ and $\ell \geq |m|$. Let us define the following three real coefficients
  \begin{align*}
    \wp_0(m,\om,\la_{m\ell}^\om) & := -2 \la^\om_{m\ell} - (\la^\om_{m\ell})^2 + 20 a (\Xi m) (\Xi\om) + 8k^{-2} (1 + \la^\om_{m\ell}) (\Xi\om)^2 - 8k^{-4}(\Xi\om)^4 \\
          & \quad - 2 a^2 \le(k^2 (-6 + \la^\om_{m\ell}) + 6 (\Xi\om)^2\ri),\\ \\
    \wp_1(m,\om,\la_{m\ell}^\om) & := 4k^{-4}\le(2 a k^4(\Xi m) + k^2 (2 + \la^\om_{m\ell}) (\Xi\om) - 2 (\Xi\om)^3\ri),\\ \\
    \wp_2(m,\om,\la_{m\ell}^\om) & := - 4 k^{-4}\big(-k^4 \la_{m\ell}^\om (1 + \la_{m\ell}^\om) (\Xi\om) +k^2 (2 + 3 \la_{m\ell}^\om) (\Xi\om)^3 - 2 (\Xi\om)^5 \\
         & \quad + a (-3 k^6 \la_{m\ell}^\om (\Xi m) + 8 k^4 (\Xi m) (\Xi\om)^2) + 2 a^2 (k^6 (3 + \la_{m\ell}^\om) (\Xi\om) - 3 k^4 (\Xi\om)^3)\big).
  \end{align*}
  For all functions $R_{[\pm2],\om}^{m\ell}$ solutions to the radial Teukolsky equations~\eqref{eq:radialTeuk} we have 
  \begin{align}\label{eq:matrixlimits}
    \begin{aligned}
    \lim_{r\to+\infty} R^{m\ell}_{[\pm2],\om,c} & = -\le(\wp_0+12iM\Xi\om\ri) \lim_{r\to+\infty} R^{m\ell}_{[\mp2],\om} + i\wp_1\lim_{r\to+\infty}\pr_{r^\ast}R^{m\ell}_{[\mp2],\om}, \\
    \lim_{r\to+\infty} \pr_{r^\ast}R^{m\ell}_{[\pm2],\om,c} & = i\wp_2\lim_{r\to+\infty} R^{m\ell}_{[\mp2],\om} - \le(\wp_0-12iM\Xi\om\ri) \lim_{r\to+\infty}\pr_{r^\ast}R^{m\ell}_{[\mp2],\om}.
    \end{aligned}
  \end{align}
\end{lemma}

\section{Teukolsky-Starobinsky conservation law and the $\om-m\om_+\neq0$ case}\label{sec:TSconslaws}
The main result of this section is the following proposition, from which Theorem~\ref{thm:main} follows (see Corollary~\ref{cor:pfmainnonstat}).
\begin{proposition}[Teukolsky-Starobinsky conservation law]\label{prop:TSconslaw}
  Let $\om\in\RRR$, $m\in\ZZZ$ and $\ell\geq |m|$. Assume that $\om-m\om_+\neq0$. Let $R_{[\pm2],\om}^{m\ell}$ be solutions to the radial Teukolsky equations~\eqref{eq:radialTeuk}, satisfying the boundary conditions at infinity~\eqref{eq:radialbdyconditions}. We have the following \emph{conservation law}
  \begin{align}
    \label{eq:TSconslaw}
    \begin{aligned}
    & \aleph(m,\om,\la^\om_{m\ell}) \big|A_{[+2],\om,\HH^+}^{m\ell}\big|^2 + \le|C(m,\om)A_{[-2],\om,\HH^+}^{m\ell} + 12i M(\Xi\om) \big(A_{[+2],\om,\HH^+}^{m\ell}\big)^\ast\ri|^2 \\
      = \; & \aleph(m,\om,\la^\om_{m\ell}) \big|A_{[-2],\om,\HH^-}^{m\ell}\big|^2 + \le|C(m,\om)A_{[+2],\om,\HH^-}^{m\ell} +  12i M(\Xi\om) \big(A_{[-2],\om,\HH^-}^{m\ell}\big)^\ast\ri|^2,
    \end{aligned}
  \end{align}
  where we refer to~\eqref{eq:radialregconditionbasisdecomp} for the definition of the horizon coefficients $A_{[\pm2],\om,\HH^\pm}^{m\ell}$ and where we recall that $\aleph(m,\om,\la^\om_{m\ell})$ is the angular Teukolsky-Starobinsky constant defined in Lemma~\ref{lem:angTS}.
\end{proposition}
\begin{corollary}[Proof of Theorem~\ref{thm:main} in the $\om-m\om_+\neq0$ case]\label{cor:pfmainnonstat}
  Let $\al^{[\pm2]}$ be regular mode solutions to the Teukolsky system~\eqref{Teusys}. Assume that $\al^{[\pm2]}$ are non-stationary with respect to the Hawking vector field, \emph{i.e.} $\mathrm{K}(\al^{[+2]}) \neq 0$ or $\mathrm{K}(\al^{[-2]}) \neq 0$. Then, $\al^{[+2]}=\al^{[-2]}=0$.
\end{corollary}
\begin{proof}[Proof of Corollary \ref{cor:pfmainnonstat}]
  Decompose $\al^{[\pm2]}$ as in Definition~\ref{def:radialdecomp}. By Lemma~\ref{lem:radTeuk} the radial quantities $R_{[\pm2],\om}^{m\ell}$ satisfy the radial Teukolsky equations~\eqref{eq:radialTeuk} and the boundary conditions~\eqref{eq:radialbdyconditions}. From the non-stationary assumptions $\mathrm{K}(\al^{[+2]})=0$ or $\mathrm{K}(\al^{[-2]})=0$, we infer that $\om-m\om_+\neq0$ and Proposition~\ref{prop:TSconslaw} applies. Now, by Lemma~\ref{lem:radTeuk} again (see~\eqref{eq:radialregconditions}), the regularity at the horizon assumptions~\eqref{eq:defreghor} for $\al^{[\pm2]}$ imply that $A_{[+2],\om,\HH^-}^{m\ell} = A_{[-2],\om,\HH^-}^{m\ell} = 0$. Plugging this in the conservation law~\eqref{eq:TSconslaw} we infer that
  \begin{align}\label{eq:conslaw2}
     \aleph(m,\om,\la^\om_{m\ell}) \big|A_{[+2],\om,\HH^+}^{m\ell}\big|^2 + \le|C(m,\om)A_{[-2],\om,\HH^+}^{m\ell} + 12i M(\Xi\om) \big(A_{[+2],\om,\HH^+}^{m\ell}\big)^\ast\ri|^2 & = 0.
  \end{align}
  By Lemma~\ref{lem:angTS}, we have that the angular Teukolsky-Starobinsky constant $\aleph(m,\om,\la^\om_{m\ell})$ is strictly positive. Moreover, from the Definition of $C(m,\om)$ from Lemma~\ref{lem:TShorasympt}, we have that if $\om-m\om_+\neq0$ then $C(m,\om)\neq0$. From these two observations, identity~\eqref{eq:conslaw2} implies that $A_{[+2],\om,\HH^+}^{m\ell}=A_{[-2],\om,\HH^+}^{m\ell}=0$ and this finishes the proof of the corollary.
\end{proof}
\begin{remark}\label{rem:crucialangTSpos}
  The key argument in the proof of Corollary~\ref{cor:pfmainnonstat} is the strict positivity of the \underline{angular} Teukolsky-Starobinsky constant $\aleph(m,\om,\la^\om_{m\ell})$. This strongly differs with the proof of mode stability in the non-superradiant asymptotically flat case which uses the positivity of the \underline{radial} Teukolsky-Starobinsky constant $\wp(m,\om,\la^\om_{m\ell}) = \aleph(m,\om,\la^\om_{m\ell}) + 144M^2(\Xi\om)^2$. Note that the positivity of the radial Teukolsky-Starobinsky constant requires a much less refined analysis as the positivity of the angular Teukolsky-Starobinsky constant, as one can directly infer that, if the radial Teukolsky-Starobinsky constant vanishes, then $\om=0$. See also Lemma~\ref{lem:angTS} for the proof of the positivity of the angular Teukolsky-Starobinsky constant obtained in the present paper.
\end{remark}
The proof of Proposition~\ref{prop:TSconslaw} relies on the following lemma.
\begin{lemma}\label{lem:wronskiansidentity}
  Let $\om\in\RRR$, $m\in\ZZZ$ and $\ell\geq |m|$ and let $R_{[\pm2],\om}^{m\ell}$ be solutions to the radial Teukolsky equations~\eqref{eq:radialTeuk}, satisfying the boundary conditions at infinity~\eqref{eq:radialbdyconditions}. Define the Wronskians
  \begin{align*}
    W_{+2} & := W\big(R_{[+2],\om}^{m\ell},\big(R_{[-2],\om,c}^{m\ell}\big)^\ast\big), & W_0 & := W\big(R_{[+2],\om}^{m\ell},R_{[-2],\om}^{m\ell}\big), &  W_{-2} & := W\big(R_{[+2],\om,c}^{m\ell},\big(R_{[-2],\om}^{m\ell}\big)^\ast\big).
  \end{align*}
  We have
  \begin{align}\label{eq:wronskiansidentity}
    W_{+2} - W_{-2} & = 24 i M(\Xi\om) W_0.
  \end{align}
\end{lemma}
\begin{proof}
  The radial functions $R_{[+2],\om}^{m\ell},(R^{m\ell}_{[-2],\om,c})^\ast$, the radial functions $R_{[+2],\om}^{m\ell},R_{[-2],\om}^{m\ell}$, and the radial functions $R_{[+2],\om,c}^{m\ell},(R_{[-2],\om}^{m\ell})^\ast$ respectively satisfy the same second order ODE (see Equations~\eqref{eq:radialTeuk}~\eqref{eq:TeukTS} and Lemma~\ref{lem:Teukeqrel}). Their Wronskians $W_{+2},W_{0},W_{-2}$ therefore satisfy the identity~\eqref{eq:wronskiansidentity} iff~\eqref{eq:wronskiansidentity} holds in the limit $r\to+\infty$. From the limits~\eqref{eq:radialbdyconditions} and~\eqref{eq:matrixlimits} we have
  \begin{align*}
    \begin{aligned}
      \lim_{r\to+\infty} W_{+2} & = \lim_{r\to+\infty} \begin{vmatrix} R_{[+2]} & R_{[-2],c}^\ast \\ \pr_{r^\ast}R_{[+2]} & \pr_{r^\ast} R_{[-2],c}^\ast \end{vmatrix} \\
      & = \lim_{r\to+\infty} \begin{vmatrix} R_{[-2]}^\ast & \le(R_{[+2],c} +12 i M (\Xi\om)\le(R_{[-2]}+R_{[+2]}^\ast\ri)\ri) \\
        -\pr_{r^\ast}R_{[-2]}^\ast & \le(-\pr_{r^\ast} R_{[+2],c}^\ast + 12 i M (\Xi\om) \le(\pr_{r^\ast}R_{[-2]}-\pr_{r^\ast}R_{[+2]}^\ast\ri)\ri) \end{vmatrix} \\
      & = \lim_{r\to+\infty} W_{-2} + 12iM(\Xi\om) \lim_{r\to+\infty} \le(\pr_{r^\ast}R^\ast_{[-2]}\le(R_{[-2]}+R_{[+2]}^\ast\ri)+R^\ast_{[-2]}\le(\pr_{r^\ast}R_{[-2]}-\pr_{r^\ast}R_{[+2]}^\ast\ri) \ri)\\
      & = \lim_{r\to+\infty} W_{-2} + 24iM(\Xi\om) \lim_{r\to+\infty}\begin{vmatrix} R_{[+2]} & R_{[-2]} \\ \pr_{r^\ast}R_{[+2]} & \pr_{r^\ast}R_{[-2]}\end{vmatrix} \\
      & = \lim_{r\to+\infty} \le(W_{-2} + 24iM(\Xi\om)W_0\ri),
    \end{aligned}
  \end{align*}
  where we dropped the indices $m,\ell,\om$ for simplicity. This finishes the proof of the lemma.
\end{proof}
We are now ready to prove Proposition~\ref{prop:TSconslaw}.
\begin{proof}[Proof of Proposition~\ref{prop:TSconslaw}]
  Using the horizon decompositions~\eqref{eq:radialregconditionbasisdecomp} and~\eqref{eq:radialregconditionsTSgeneral} and the symmetries from Lemma~\ref{lem:Teukeqrel}, we have
  \begin{align*}
    \begin{aligned}
      W_{+2} & = W\bigg(A^{m\ell}_{[+2],\om,\HH^+}R_{[+2],\om,\HH^+}^{m\ell}+A_{[+2],\om,\HH^-}^{m\ell}R^{m\ell}_{[-2],\om,\HH^-},\\
      & \quad \quad \quad \wp(m,\om,\la^\om_{m\ell}) C(m,\om)^{-1}\big(A_{[+2],\HH^+}\big)^\ast R^{m\ell}_{[-2],\om,\HH^+}+C(m,\om)^\ast \big(A_{[+2],\om,\HH^-}^{m\ell}\big)^\ast R^{m\ell}_{[-2],\om,\HH^-}\bigg) \\
      & = \wp(m,\om,\la^\om_{m\ell}) C(m,\om)^{-1}\big|A_{[+2],\om,\HH^+}^{m\ell}\big|^2 W\le(R_{[+2],\om,\HH^+}^{m\ell},R^{m\ell}_{[-2],\om,\HH^+}\ri) \\
      & \quad + C(m,\om)^\ast\big|A_{[+2],\om,\HH^-}^{m\ell}\big|^2 W\le(R_{[+2],\om,\HH^-}^{m\ell},R_{[-2],\om,\HH^-}^{m\ell}\ri).
    \end{aligned}
  \end{align*}
  A direct computation gives
  \begin{align*}
    \begin{aligned}
      W\le(R_{[+2],\om,\HH^+}^{m\ell},R^{m\ell}_{[-2],\om,\HH^+}\ri) = -W\le(R_{[+2],\om,\HH^-}^{m\ell},R_{[-2],\om,\HH^-}^{m\ell}\ri) = 2\frac{\xi}{r_+^2+a^2}\De^2F(r),
    \end{aligned}
  \end{align*}
  where $F:[r_+,+\infty)\to\CCC$ is a smooth function with $F(r_+)=1$. Thus, we infer that
  \begin{align}\label{eq:limwronskianshorizon1}
    \De^{-2}W_{+2}\bigg|_{r=r_+} & = 2\frac{\xi}{r_+^2+a^2} \wp(m,\om,\la^\om_{m\ell}) C(m,\om)^{-1}\big|A_{[+2],\om,\HH^+}^{m\ell}\big|^2 -2\frac{\xi}{r_+^2+a^2} C(m,\om)^\ast\big|A_{[+2],\om,\HH^-}^{m\ell}\big|^2.
  \end{align}
  Analogous computations give
  \begin{align}\label{eq:limwronskianshorizon2}
    \begin{aligned}
      \De^{-2}W_{0}\bigg|_{r=r_+} & = 2 \frac{\xi}{r_+^2+a^2}A_{[+2],\om,\HH^+}^{m\ell}A_{[-2],\om,\HH^+}^{m\ell} - 2 \frac{\xi}{r_+^2+a^2}A_{[+2],\om,\HH^-}^{m\ell}A_{[-2],\om,\HH^-}^{m\ell},\\ \\
      \De^{-2}W_{-2}\bigg|_{r=r_+} & = 2\frac{\xi^\ast}{r_+^2+a^2} C(m,\om)\big|A_{[-2],\om,\HH^+}^{m\ell}\big|^2- 2\frac{\xi^\ast}{r_+^2+a^2} \wp(m,\om,\la^\om_{m\ell}) (C(m,\om)^\ast)^{-1}\big|A_{[-2],\om,\HH^-}^{m\ell}\big|^2.
    \end{aligned}
  \end{align}

  Plugging~\eqref{eq:limwronskianshorizon1} and~\eqref{eq:limwronskianshorizon2} in the identity~\eqref{eq:wronskiansidentity}, we obtain
  \begin{align*}
    & \wp C^{-1}\xi |A_{[+2],\HH^+}|^2 -C^\ast\xi|A_{[+2],\HH^-}|^2 - C \xi^\ast |A_{[-2],\HH^+}|^2 + \wp (C^{\ast})^{-1}\xi^\ast|A_{[-2],\HH^-}|^2 \\
    = & \; 24iM(\Xi\om) \xi \le(A_{[+2],\HH^+}A_{[-2],\HH^+} - A_{[+2],\HH^-}A_{[-2],\HH^-}\ri),
  \end{align*}
  where here and in the rest of this proof we dropped the indices and arguments $m,\ell,\om,\la_{m\ell}^{\om}$ for simplicity. Using the expression of $C$, we rewrite this identity as
  \begin{align}\label{eq:wronskhorid}
    \begin{aligned}
      & i\wp \widetilde{C}^{-1}(\Xi(\om-m\om_+))^{-1} |A_{[+2],\HH^+}|^2 + i (\Xi(\om-m\om_+)) \widetilde{C}|\xi|^2|A_{[-2],\HH^+}|^2 - 24iM(\Xi\om) \xi A_{[+2],\HH^+}A_{[-2],\HH^+} \\
      = & \; i\wp \widetilde{C}^{-1}(\Xi(\om-m\om_+))^{-1} |A_{[-2],\HH^-}|^2 + i (\Xi(\om-m\om_+)) \widetilde{C}|\xi|^2|A_{[+2],\HH^-}|^2 - 24iM(\Xi\om) \xi A_{[+2],\HH^-}A_{[-2],\HH^-}.
    \end{aligned}
  \end{align}
  Using that $\widetilde{C}>0$ and $\xi \neq 0$ we define $B_{[\pm2],\HH^\pm}$ by
  \begin{align*}
    A_{[+2],\HH^+} & =: \le|\Xi(\om-m\om_+)\ri|^{+1/2} \widetilde{C}^{+1/2} B_{[+2],\HH^+}, & A_{[-2],\HH^+} & =: \xi^{-1}\le|\Xi(\om-m\om_+)\ri|^{-1/2} \widetilde{C}^{-1/2} B_{[-2],\HH^+} \\
    A_{[-2],\HH^-} & =: \le|\Xi(\om-m\om_+)\ri|^{+1/2} \widetilde{C}^{+1/2} B_{[-2],\HH^-}, & A_{[+2],\HH^-} & =: \xi^{-1}\le|\Xi(\om-m\om_+)\ri|^{-1/2} \widetilde{C}^{-1/2} B_{[+2],\HH^-}.
  \end{align*}
  Plugging these definitions in~\eqref{eq:wronskhorid}, we have
  \begin{align*}
    & \wp|B_{[+2],\HH^+}|^2 + |B_{[-2],\HH^+}|^2 - 24 M (\Xi\om) B_{[+2],\HH^+}B_{[-2],\HH^+} \\
    = & \; \wp|B_{[-2],\HH^-}|^2 + |B_{[+2],\HH^-}|^2 - 24 M (\Xi\om) B_{[-2],\HH^-}B_{[+2],\HH^-},
  \end{align*}
  if $\om-m\om_+>0$, and
  \begin{align*}
    & \wp|B_{[+2],\HH^+}|^2 + |B_{[-2],\HH^+}|^2 + 24 M (\Xi\om) B_{[+2],\HH^+}B_{[-2],\HH^+} \\
    = & \; \wp|B_{[-2],\HH^-}|^2 + |B_{[+2],\HH^-}|^2 + 24 M (\Xi\om) B_{[-2],\HH^-}B_{[+2],\HH^-},
  \end{align*}
  if $\om-m\om_+<0$. Using that $\wp=144M^2(\Xi\om)^2+\aleph$, we rewrite the above as
  \begin{align*}
    & \aleph|B_{[+2],\HH^+}|^2 + 144 M^2(\Xi\om)^2 |B_{[+2],\HH^+}|^2 + |B_{[-2],\HH^+}|^2 \mp 24 M (\Xi\om) B_{[+2],\HH^+}B_{[-2],\HH^+} \\
    & = \aleph|B_{[-2],\HH^-}|^2 + 144 M^2(\Xi\om)^2 |B_{[-2],\HH^-}|^2 + |B_{[+2],\HH^-}|^2 \mp 24 M (\Xi\om) B_{[-2],\HH^-}B_{[+2],\HH^-}.
  \end{align*}
  Taking the real part of the above, and using that $|z_1+z_2|^2=|z_1|^2+|z_2|^2+2\Re(z_1^\ast z_2)$ for all $z_1,z_2\in\CCC$, we infer
  \begin{align*}
    \aleph |B_{[+2],\HH^+}|^2 + \le|B_{[-2],\HH^+} \mp 12 M(\Xi\om) B_{[+2],\HH^+}^\ast\ri|^2  & = \aleph |B_{[-2],\HH^-}|^2 + \le|B_{[+2],\HH^-} \mp 12 M(\Xi\om) B_{[-2],\HH^-}^\ast\ri|^2.
  \end{align*}
  Re-plugging the definition of the quantities $B_{[\pm2],\HH^\pm}$ in the above, we obtain~\eqref{eq:TSconslaw} and this finishes the proof of the proposition.
\end{proof}

\section{The stationary case $\om-m\om_+=0$}\label{sec:stat}
The second part of Theorem~\ref{thm:main} splits into the following two propositions.
\begin{proposition}[The $m=0$ case]\label{prop:statm0}
  Let $m=0$ and $\om = m\om_+ = 0$. For all radial functions $R_{[\pm2],m}^{m\ell}$, solutions to the Teukolsky equations~\eqref{eq:radialTeuk} satisfying the boundary conditions at infinity~\eqref{eq:radialbdyconditions} and the regularity conditions~\eqref{eq:radialregstat} at the horizon, we have $R_{[+2],m}^{m\ell}=R_{[-2],m}^{m\ell}=0$.
\end{proposition}
\begin{proposition}[The general $m$ case and the Hawking-Reall bound under the restrictions~\eqref{est:slowrotmore}]\label{prop:stat}
  If the black hole parameters are such that the Hawking-Reall bound~\eqref{est:HRboundoriginal} and either assumption~\eqref{est:assumboundX} or~\eqref{est:farHR} are satisfied, then the following holds. For all $m\in\ZZZ$ and for all radial functions $R_{[\pm2],\om}^{m\ell}$, solutions to the Teukolsky equations~\eqref{eq:radialTeuk} with $\om=m\om_+$ and satisfying the boundary conditions at infinity~\eqref{eq:radialbdyconditions} and the regularity conditions~\eqref{eq:radialregstat} at the horizon, we have $R_{[+2],\om}^{m\ell}=R_{[-2],\om}^{m\ell}=0$.
\end{proposition}
\begin{remark}\label{rem:sharpHR}
  The positivity of the potentials crucially used in the proof of Proposition~\ref{prop:stat} dramatically fails to hold if the Hawking-Reall bound~\eqref{est:HRboundoriginal} is violated. In that case, the potentials can indeed be made arbitrarily negative provided that the azimuthal number $|m|$ is sufficiently large (see Remarks~\ref{rem:angularestimatessharpde1} and~\ref{rem:angularestimatessharpde2}). In the case of the wave equation, the same feature was the key to construct stationary (and growing) mode solutions in~\cite{Dol17}.
\end{remark}

The proof of Propositions~\ref{prop:statm0} and~\ref{prop:stat} is postponed to Sections~\ref{sec:proofpropstatm0} and~\ref{sec:proofpropstat} respectively. It is based on an energy identity and uses estimates on the angular eigenvalues which are respectively obtained in Sections~\ref{sec:energy} and~\ref{sec:angularestimates}. In some specific case of the proof of Proposition~\ref{prop:stat}, we will also need Hardy estimates which are obtained in the Section~\ref{sec:hardy}.
\subsection{Energy identity}\label{sec:energy}
We have the following energy identity.
\begin{lemma}[Energy identity]\label{lem:energy}
  For all radial functions $R=R_{[\pm2],\om}^{m\ell}$, solutions to the Teukolsky equations~\eqref{eq:radialTeuk} with $\om=m\om_+$ satisfying the boundary conditions at infinity~\eqref{eq:radialbdyconditions} and the regularity conditions~\eqref{eq:radialregstat} at the horizon, we have
  \begin{align}\label{est:energy}
    \begin{aligned}
      \int_{-\infty}^{\pi/2} \le(|\pr_{r^\ast}\widetilde{R}|^2 + V_{\mathrm{stat}}^m[\lat] |\widetilde{R}|^2 +\frac{\De}{(r^2+a^2)^2}\le(\lat_{m\ell}^{m\om_+}-\lat\ri) |\widetilde{R}|^2 \ri) \d r^\ast = 0, 
    \end{aligned}
  \end{align}
  for all $\lat\in\RRR$, where
  \begin{align*}
    \widetilde{R} & := \frac{(r^2+a^2)^2}{\De} R,
  \end{align*}
  and where
  \begin{align*}
    V_{\mathrm{stat}}^{m}[\lat] & := \frac{\De}{(r^2+a^2)^2}\lat + \le(\frac{\Xi\om_+m}{k}\ri)^2\frac{\De-k^2(r-r_+)^2(r+r_+)^2}{(r^2+a^2)^2} \\
                                & \quad + \frac{(\pr_r\De)^2}{(r^2 + a^2)^2} + \frac{\De}{(r^2+a^2)^2}\le(2+a^2k^2-6k^2r^2-\pr_r^2\De\ri) +r\frac{\De\pr_r\De}{(r^2+a^2)^{3}} - (2r^2-a^2)\frac{\De^2}{(r^2+a^2)^4},
  \end{align*}
  and
  \begin{align*}
    \lat_{m\ell}^{m\om_+} & := \la_{m\ell}^{m\om_+}-2-a^2k^2-\le(\frac{\Xi \om_+ m }{k}\ri)^2.
  \end{align*}
\end{lemma}
To prove lemma~\ref{lem:energy}, we need the following preliminary lemma.
\begin{lemma}
  For all radial functions $R_{[+2],\om}^{m\ell}, R_{[+2],\om}^{m\ell}$, solutions to the Teukolsky equations~\eqref{eq:radialTeuk} with $\om=m\om_+$ satisfying the boundary conditions at infinity~\eqref{eq:radialbdyconditions} and the regularity conditions~\eqref{eq:radialregstat} at the horizon, we have
  \begin{align}\label{eq:bdycondim}
    \begin{aligned}
      R_{[+2],\om}^{m\ell}\pr_{r^\ast}\big(R_{[+2],\om}^{m\ell}\big)^\ast+\big(R_{[+2],\om}^{m\ell}\big)^\ast\pr_{r^\ast}R_{[+2],\om}^{m\ell} \xrightarrow{r\to+\infty} 0, \\
      R_{[-2],\om}^{m\ell}\pr_{r^\ast}\big(R_{[-2],\om}^{m\ell}\big)^\ast+\big(R_{[-2],\om}^{m\ell}\big)^\ast\pr_{r^\ast}R_{[-2],\om}^{m\ell} \xrightarrow{r\to+\infty} 0.
    \end{aligned}
  \end{align}
\end{lemma}
\begin{proof}
  From the regularity conditions at the horizon~\eqref{eq:radialregstat}, $R_{[+2]}=R_{[+2],\om}^{m\ell}$ and $R_{[-2]}=R_{[-2],\om}^{m\ell}$ are on the same branch of the Frobenius expansion of the ODE~\eqref{eq:radialTeuk} when $r^\ast\to-\infty$. If $R_{[+2]}=0$ or $R_{[-2]}=0$ then by a unique-continuation argument for the ODE~\eqref{eq:radialTeuk} using the boundary conditions~\eqref{eq:radialbdyconditions}, we have $R_{[+2]}=R_{[-2]}=0$ and~\eqref{eq:bdycondim} trivially holds. If else, there exists $\kappa\in\CCC^\ast$ such that $R_{[+2]}(r)=\kappa R_{[-2]}(r)$. The boundary conditions at infinity~\eqref{eq:radialbdyconditions} imply that
  \begin{align}\label{eq:bdycondtied}
    \begin{aligned}
      R_{[+2]}-\kappa^\ast R_{[+2]}^\ast & \xrightarrow{r\to+\infty} 0, & \pr_{r^\ast}R_{[+2]}+\kappa^\ast \pr_{r^\ast}R_{[+2]}^\ast & \xrightarrow{r\to+\infty} 0.  
    \end{aligned}
  \end{align}
  From~\eqref{eq:bdycondtied}, we deduce that
  \begin{align*}
    R_{[+2]}\pr_{r^\ast}R_{[+2]}^\ast +R_{[+2]}^\ast\pr_{r^\ast}R_{[+2]} & \xrightarrow{r\to+\infty} \kappa^\ast R_{[+2]}^\ast\pr_{r^\ast}R_{[+2]}^\ast - \kappa^\ast R_{[+2]}^\ast\pr_{r^\ast}R_{[+2]}^\ast = 0.
  \end{align*}
  The same argument also gives the same limit in the $[-2]$ case and this finishes the proof of the lemma.
\end{proof}
\begin{proof}[Proof of Lemma~\ref{lem:energy}]
  Using the expressions and notations of Lemma~\ref{lem:Teukeqrel}, we have
  \begin{align}\label{eq:radialTeukstatrenorm}
    \begin{aligned}
      0 & = \frac{(r^2+a^2)^2}{\De}\widetilde{\mathfrak{I}}_{[+2]}^{m,\om}[\la^\om_{m\ell}]\le(\frac{\De}{(r^2+a^2)^2} \widetilde{R}\ri) \\
      & = \pr_{r^\ast}^2\widetilde{R} -\le(V_{\mathrm{stat}}^{m}[\lat^\om_{m\ell}]-i\Im\le(V^{m,\om}[\la^\om_{m\ell}]\ri)\ri)\widetilde{R}.
    \end{aligned}
  \end{align}
  From~\eqref{eq:radialTeukstatrenorm} we infer
  \begin{align}\label{eq:radialTeukstatrenorm2}
    \begin{aligned}
      0 & = \pr_{r^\ast}\le(\widetilde{R}\pr_{r^\ast}\widetilde{R}^\ast + \widetilde{R}^\ast\pr_{r^\ast}\widetilde{R}\ri) - 2|\pr_{r^\ast}\widetilde{R}|^2 -2 V_{\mathrm{stat}}^{m}[\lat^\om_{m\ell}] |\widetilde{R}|^2.
    \end{aligned}
  \end{align}
  Integrating~\eqref{eq:radialTeukstatrenorm2}, using the conditions at the horizon~\eqref{eq:radialregstat} and the limit at infinity~\eqref{eq:bdycondim}, we obtain the desired~\eqref{est:energy} and this finishes the proof of the lemma.
\end{proof}
\begin{remark}\label{rem:angularestimatessharpde1}
  From a direct computation, one can check that the renormalised angular eigenvalue $\lat$ is the leading order term of the stationary potential at infinity, \emph{i.e.}
  \begin{align}\label{eq:limitVstat}
    V_{\mathrm{stat}}^{m} [\lat=\lat_{m\ell}^{m\om_+}] & \xrightarrow{r\to +\infty} k^2 \lat_{m\ell}^{m\om_+}.   
  \end{align}
  The key to establishing the positivity of the potential (up to using Hardy estimates) is therefore to obtain a suitable lower bound on $\lat_{m\ell}^{m\om_+}$. See Section~\ref{sec:angularestimates} and Remark~\ref{rem:angularestimatessharpde2}. 
\end{remark}

\subsection{Hardy estimate}\label{sec:hardy}
The following Hardy estimate is used in the proof of Proposition~\ref{prop:stat}. It is designed to compensate the potentially negative leading coefficient of the potential while keeping the lower coefficients positive (see also Remark~\ref{rem:angularestimatessharpde1}).
\begin{lemma}[Hardy estimate]\label{lem:energyandHardy}
  For all radial functions $\widetilde{R}$ such that $R:=\frac{(r^2+a^2)}{\De}\widetilde{R}$ satisfies the boundary conditions at infinity~\eqref{eq:radialbdyconditions} and the regularity conditions~\eqref{eq:radialregstat} at the horizon, we have
  \begin{align}\label{est:Hardy}
    \begin{aligned}
      \int_{-\infty}^{\pi/2}|\pr_{r^\ast}\widetilde{R}|^2 \,\d r^\ast & \geq \int_{-\infty}^{\pi/2}V_{\mathrm{Hardy}}[\lat]|\widetilde{R}|^2 \,\d r^\ast.
    \end{aligned}
  \end{align}
  for all $\lat\in\RRR$, where
  \begin{align*}
    V_{\mathrm{Hardy}}[\lat] & := \lat\frac{\De}{r^2+a^2}\frac{\d}{\d r}\le(\frac{r-r_+}{r^2+a^2}\ri)  - \lat^2\le(\frac{r-r_+}{r^2+a^2}\ri)^2.
  \end{align*}
\end{lemma}
\begin{proof}
  Let us define
  \begin{align*}
    y(r^\ast) & := \lat\frac{r-r_+}{r^2+a^2}, & f(r^\ast) & := \exp\le(\int_{-\infty}^{r^\ast}y(r^{\ast,'})\,\d r^{\ast,'}\ri).
  \end{align*}
  We have
  \begin{align*}
    0 & \leq \int_{-\infty}^{\pi/2}\le|f^{-1}\pr_{r^\ast}\le(f\widetilde{R}_{[+2]}\ri)\ri|^2 \,\d r^\ast \\
      & = \int_{-\infty}^{\pi/2}|\pr_{r^\ast}\widetilde{R}_{[+2]}|^2 \,\d r^\ast + \int_{-\infty}^{\pi/2}\big|\pr_{r^\ast}(\log f)\widetilde{R}_{[+2]}\big|^2 \,\d r^\ast \\
      & \quad + \int_{-\infty}^{\pi/2}\pr_{r^\ast}(\log f)\le(\widetilde{R}_{[+2]}^\ast\pr_{r^\ast}\widetilde{R}_{[+2]} + \widetilde{R}_{[+2]}\pr_{r^\ast}\widetilde{R}_{[+2]}^\ast\ri)\,\d r^\ast \\
      & = \int_{-\infty}^{\pi/2}|\pr_{r^\ast}\widetilde{R}_{[+2]}|^2 \,\d r^\ast + \int_{-\infty}^{\pi/2}\le(-\pr_{r^\ast}^2(\log f) + \le(\pr_{r^\ast}(\log f)\ri)^2\ri)|\widetilde{R}_{[+2]}|^2 \,\d r^\ast +\le[\pr_{r^\ast}(\log f) |\widetilde{R}_{[+2]}|^2\ri]_{-\infty}^{\pi/2}.
  \end{align*}
  Using that $y(r^\ast)\to 0$ when $r^\ast \to \pi/2$ and $r^\ast\to-\infty$, and using the limits given by the conditions~\eqref{eq:radialbdyconditions} and~\eqref{eq:radialregstat}, the boundary terms in the above computation vanish and the bound~\eqref{est:Hardy} follows.
\end{proof}

\subsection{Estimates for the angular eigenvalues}\label{sec:angularestimates}
Define the following two parameters
\begin{align*}
  X & := ak, & a\om_+ & =: \frac{\de X}{1+X}.
\end{align*}
We have the following bounds on the angular eigenvalues.
\begin{lemma}\label{lem:lambda}
  Let $m\in\ZZZ$. For all admissible Kerr-adS black hole parameters $(M,a,k)$ and for all $\ell\geq 2$, the following holds.
  \begin{enumerate}
  \item For $m=0$, we have
    \begin{align}\label{est:lam0}
      \lat_{m\ell}^{m\om_+} & \geq 0.
    \end{align}
  \item\label{item:lambdaHR} If the Hawking-Reall bound~\eqref{est:HRboundoriginal} holds, we have for $|m|=1$
    \begin{align}
      \label{est:lamgeneralm1}
      \lat_{m\ell}^{m\om_+} & \geq -(1+X)^2 -2X(1-X),
    \end{align}
    and, for $|m|\geq 2$,
    \begin{align}
      \label{est:lamgeneral}
      \lat_{m\ell}^{m\om_+} & \geq -4(1+X)^2 + 4\Xi(1-\de^2)(1-X)^2.
    \end{align}
  \item\label{item:lambdaHRplus} If the Hawking-Reall bound~\eqref{est:HRboundoriginal} and assumption~\eqref{est:assumboundX} hold, we have for $|m|=1$
    \begin{align}
      \label{est:lampos}
      \lat_{m\ell}^{m\om_+} & \geq 0,
    \end{align}
    and, for $|m| \geq 2$, 
    \begin{align}
      \label{est:lam234}
      \lat^{m\om_+}_{m\ell} & \geq -\frac{5\de}{2}.
    \end{align}
  \end{enumerate}
\end{lemma}
\begin{remark}
  As it will be clear in the proof of Item~\ref{item:lambdaHRplus}, the positive bound~\eqref{est:lampos} also holds for $|m|\geq 5$. However, in that case, the bound~\eqref{est:lam234} is sufficient for the proof of Proposition~\ref{prop:stat}. Also the bound~\eqref{est:lam234} is not sharp but sufficient for the proof of Proposition~\ref{prop:stat}.
\end{remark}
Before turning to the proof of Lemma~\ref{lem:lambda}, we have the following preliminary estimate.
\begin{lemma}\label{lem:prelimlambda}
  Let $m\in\ZZZ$. For all admissible Kerr-adS black hole parameters $(M,a,k)$ and for all $\ell\geq 2$, we have
  \begin{align}\label{est:minGtilde}
    \begin{aligned}
      \lat_{m\ell}^{m\om_+} & \geq \Xi \int_{0}^\pi|\pr_\varth S_{m\ell}^{m\om_+}|^2\sin\varth\d\varth + \int_0^\pi F(\varth)|S_{m\ell}^{m\om_+}|^2\sin\varth\d\varth,
    \end{aligned}
  \end{align}
  where
  \begin{align*}
    F(\varth) & := \frac{\Xi^2}{\Deth\sin^2\varth}\le(4m\cos\varth+(4+m^2)\cos^2\varth\ri) + \frac{\Xi}{\Deth}X(1+X)4m\cos\varth \\
              & \quad + \frac{\Xi}{\Deth}(1 -\de)(1-X)\le((1+X +\de (1 - X)) m^2-4mX\cos\varth\ri).
  \end{align*}
\end{lemma}
\begin{proof}
  We drop the indices for simplicity. Multiplying~\eqref{eq:defla} by $S\sin\varth$, integrating over $\varth \in (0,\pi)$ and using the regularity conditions~\eqref{est:limitsStheta} at $\varth=0,\pi$, we have
  \begin{align}\label{eq:lavar}
    \begin{aligned}
    \la\int_0^\pi|S|^2\sin\varth\d\varth & = \int_{0}^\pi -\sqrt{\De_\varth}\LL_{-1}^\dg\sqrt{\De_\varth}\LL_2S+\le(6a\Xi\om\cos\varth-6a^2k^2\cos^2\varth\ri)|S|^2 \sin\varth\d\varth \\
    & = \int_0^\pi \le(\De_\varth|\pr_\varth S|^2 + G(\varth) |S|^2 \ri)\sin\varth\d\varth,
    \end{aligned}
  \end{align}
  where
  \begin{align*}
    G(\varth) & := 2+a^2k^2 + \le(\frac{\Xi\om_+m}{k}\ri)^2 -a^2k^2\cos(2\theta) + \frac{\Xi^2}{\Deth\sin^2\varth} \le(4m\cos\varth + \le(4 + m^2\ri)\cos^2\varth\ri) \\
              & \quad + \frac{\Xi}{\Deth}X(1+X)4m\cos\varth + \frac{\Xi}{\Deth}(1 -\de)(1-X)\le((1+X +\de (1 - X)) m^2-4mX\cos\varth\ri).
  \end{align*}
  
  We have
  \begin{align}\label{est:X1prvarthS}
    \begin{aligned}
      \int_{0}^\pi\De_\varth|\pr_\varth S|^2\sin\varth & = \int_{0}^\pi \le(\Xi|\pr_\varth S|^2+X^2\sin^2\varth|\pr_\varth S|^2\ri)\sin\varth\d\varth\\
      & =  \Xi \int_{0}^\pi|\pr_\varth S|^2\sin\varth\d\varth \\
      & \quad + X^2\int_0^\pi \le(\le|\pr_\varth \le(\sin\varth S\ri)\ri|^2 -2\sin\varth\cos\varth S\pr_\varth S -\cos^2\varth|S|^2\ri) \sin\varth\d\varth \\
      & \geq \Xi \int_{0}^\pi|\pr_\varth S|^2\sin\varth\d\varth +  X^2\int_{0}^\pi\cos(2\varth)|S|^2\sin\varth\d\varth.
    \end{aligned}
  \end{align}
  Combining~\eqref{eq:lavar} and~\eqref{est:X1prvarthS} we obtain~\eqref{est:minGtilde} and this finishes the proof of the lemma.
\end{proof}
\begin{proof}[Proof of Lemma~\ref{lem:lambda}]
  In the $m=0$ case, the proof of~\eqref{est:lam0} follows directly from~\eqref{est:minGtilde}. We turn to the proof of Item~\ref{item:lambdaHR} and we assume that the Hawking-Reall bound~\eqref{est:HRboundoriginal} holds. This bound together with the definition of the admissible Kerr-adS black hole parameter (see Definition~\ref{def:admissibleKadSparameters}), rewrites as
  \begin{align}\label{est:assumHRbound}
    0 \leq \de \leq 1.
  \end{align}
  Note also that that same definition imposes
  \begin{align}\label{est:Xnormal}
    0 \leq X < 1.
  \end{align}
  For $m=1$ (the $m=-1$ case is obtained along the same lines), we have
  \begin{align}\label{est:Fvarthm1}
    \begin{aligned}
      F(\varth) & = \frac{\Xi^2}{\Deth\sin^2\varth} \le(2\cos\varth + 1 + \frac{X\sin^2\varth}{1-X}\ri)^2 - (1+X)^2 \\
      & \quad + \frac{\Xi}{\Deth}(1 -\de)(1-X)\le((1+X +\de (1 - X))-4X\cos\varth\ri) \\
      & \geq - (1+X)^2 + \frac{\Xi}{\Deth}(1 -\de)(1-X)\le((1+X +\de (1 - X))-4X\cos\varth\ri)\\
      & \geq -(1+X)^2 -2X(1-X),
    \end{aligned}
  \end{align}
  where in the last line we used~\eqref{est:assumHRbound} and~\eqref{est:Xnormal}. For $m\geq 2$ (the $m\leq -2$ case is obtained along the same lines), we have
  \begin{align}\label{est:Fvarthmgeneral}
    \begin{aligned}
      F(\varth) & = \frac{\Xi^2}{\Deth\sin^2\varth} \le(m\cos\varth + 2 + \frac{2X\sin^2\varth}{1-X}\ri)^2 - 4(1+X)^2 \\
      & \quad + \frac{\Xi}{\Deth}(1 -\de)(1-X)\le((1+X +\de (1 - X))m^2-4mX\cos\varth\ri)\\
      & \geq -4(1+X)^2 + \frac{\Xi}{\Deth}(1 -\de)(1-X)\le(4(1+X +\de (1 - X))-8X\cos\varth\ri) \\
      & \geq -4(1+X)^2 +4\Xi(1 -\de^2)(1-X)^2,
    \end{aligned}
  \end{align}
  where in the two last lines we used~\eqref{est:assumHRbound} and~\eqref{est:Xnormal}. Combining~\eqref{est:Fvarthm1} and~\eqref{est:Fvarthmgeneral} with~\eqref{est:minGtilde} this yields~\eqref{est:lamgeneralm1} and~\eqref{est:lamgeneral} respectively and finishes the proof of Item~\ref{item:lambdaHR}.\\
  
  We now turn to the proof of Item~\ref{item:lambdaHRplus} and assume that~\eqref{est:assumboundX} holds. The definition of admissible Kerr-adS black hole parameters allows us to rewrites this as
  \begin{align}\label{est:assumboundXbis}
    0 & \leq X \leq \frac{1}{20}.
  \end{align}
  Define
  \begin{align*}
    \widetilde{F}(\varth) & := \frac{\Xi}{\Deth}X(1+X)4m\cos\varth + \frac{\Xi}{\Deth}(1 -\de)(1-X)\le((1+X +\de (1 - X)) m^2-4mX\cos\varth\ri).
  \end{align*}
  Let us consider $m=1$ (the $m=-1$ case is obtained along the same lines). Using~\eqref{est:assumHRbound} and~\eqref{est:Xnormal}, we have
  \begin{align}\label{est:Gtilde1}
    \widetilde{F}(\varth) & \geq -4X(1+X) - 2X(1-X).
  \end{align}
  We have
  \begin{align*}
    & \Xi \int_{0}^\pi|\pr_\varth S|^2\sin\varth\d\varth + \int_0^\pi \frac{\Xi^2}{\Deth\sin^2\varth}\le(4m\cos\varth+(4+m^2)\cos^2\varth\ri)|S|^2\sin\varth\d\varth \\
    & \geq \Xi^2\le(\int_{0}^\pi|\pr_\varth S|^2\sin\varth\d\varth + \int_0^\pi \frac{1}{\sin^2\varth}\le(4m\cos\varth+(4+m^2)\cos^2\varth\ri)|S|^2\sin\varth\d\varth\ri) \\
    & = \Xi^2\le(1 + \int_{0}^\pi\le(\sin\varth\pr_{\varth}(\sin^{-2}\varth S)-2\sin^{-1}\varth (\sin^{-2}\varth S)\ri)^2\sin^{3}\varth\,\d\varth\ri) \\
    & \geq \Xi^2,
  \end{align*}
  where the third line is obtained by integration by parts, using the limits at $\varth\to 0,\pi$ given by~\eqref{est:limitsStheta}. Plugging the above estimate and~\eqref{est:Gtilde1} into~\eqref{est:minGtilde}, we obtain that
  \begin{align*}
    \lat & \geq (1-X^2)^2-4X(1+X)-2X(1-X).
  \end{align*}
  This polynomial is positive under the bound~\eqref{est:assumboundXbis}, and estimate~\eqref{est:lampos} follows.\\
  
  Let us consider $m \geq 2$. Using~\eqref{est:assumHRbound} and~\eqref{est:Xnormal}, we have
  \begin{align}
    \label{est:Gtildem2}
    \begin{aligned}
      \widetilde{F}(\varth) & \geq -4X(1+X)m + (1-\de)(1-X)\le((1+X+\de(1-X))m^2-4mX\ri).
    \end{aligned}
  \end{align}
  We have
  \begin{align*}
    & \Xi \int_{0}^\pi|\pr_\varth S|^2\sin\varth\d\varth + \int_0^\pi \frac{\Xi^2}{\Deth\sin^2\varth}\le(4m\cos\varth+(4+m^2)\cos^2\varth\ri)|S|^2\sin\varth\d\varth \\
    & \geq \Xi^2\le(\int_{0}^\pi|\pr_\varth S|^2\sin\varth\d\varth + \int_0^\pi \frac{1}{\sin^2\varth}\le(4m\cos\varth+(4+m^2)\cos^2\varth\ri)|S|^2\sin\varth\d\varth\ri) \\
    & = \Xi^2\le(m-4 + \int_{0}^\pi\le(\sin\varth\pr_{\varth}(\sin^{-m}\varth S)-2\sin^{-1}\varth (\sin^{-m}\varth S)\ri)^2\sin^{2m}\varth\,\d\varth\ri) \\
    & \geq \Xi^2(m-4),
  \end{align*}
  where the third line is obtained by integration by parts, using the limits at $\varth\to 0,\pi$ given by~\eqref{est:limitsStheta}. Plugging the above estimate and~\eqref{est:Gtildem2} into~\eqref{est:minGtilde}, we obtain
  \begin{align}\label{eq:latmm}
    \begin{aligned}
      \lat & \geq (1-X^2)^2(m-4) -4X(1+X)m + (1-\de)(1-X)\le((1+X+\de(1-X))m^2-4mX\ri) \\
      & = \le((1-\de)(1-X)(1+X+\de(1-X))\ri)m^2 + \le((1-X^2)^2 -4X(1+X) - 4X(1-\de)(1-X)\ri)m \\
      & \quad -4(1-X^2)^2.
    \end{aligned}
  \end{align}
  Under the bounds~\eqref{est:assumHRbound} and~\eqref{est:assumboundXbis}, we have
  \begin{align}\label{est:coeffmangular}
    \begin{aligned}
      (1-\de)(1-X)(1+X+\de(1-X)) & \geq 0, & (1-X^2)^2 -4X(1+X) - 4X(1-\de)(1-X) & \geq 0.
    \end{aligned}
  \end{align}
  With the bounds~\eqref{est:coeffmangular}, the right-hand side of~\eqref{eq:latmm} is a polynomial of order $2$ (or $1$) increasing in $m\geq2$. We thus have
  \begin{align}\label{est:latmm2}
    \begin{aligned}
      \lat & \geq 4\le((1-\de)(1-X)(1+X+\de(1-X))\ri) + 2\le((1-X^2)^2 -4X(1+X) - 4X(1-\de)(1-X)\ri) \\
      & \quad -4(1-X^2)^2 \\
      & = 4(1-X)^2(1-\de^2)-2(1-X^2)^2 -8X(1+X). 
    \end{aligned}
  \end{align}
  Under the bound~\eqref{est:assumboundXbis}, we have
  \begin{align}\label{est:coeffmangular2}
    \begin{aligned}
      4(1-X)^2 & \geq 5/2, &  -2(1-X^2)^2 -8X(1+X) & \geq - 5/2.
    \end{aligned}
  \end{align}
  Using~\eqref{est:coeffmangular2} in~\eqref{est:latmm2}, we deduce
  \begin{align*}
    \lat & \geq \frac{5}{2} \le( (1-\de^2) - 1 \ri) \geq -\frac{5\de}{2},
  \end{align*}
  which concludes the proof of~\eqref{est:lam234} and of the lemma.
\end{proof}

\begin{remark}\label{rem:angularestimatessharpde2}
  Using the notations of the proof of Lemmas~\ref{lem:lambda} and~\ref{lem:prelimlambda}, we have by the min-max principle that
  \begin{align*}
    \lat_{mm}^{m\om_+} & = \min_{\norm{S}=1} \le(\int_0^\pi \le(\Deth (\pr_\varth S)^2 + G(\varth) S^2\ri)\sin\varth\,\d\varth \ri) \\
                       & \leq \min_{\norm{S}=1} \le(\int_0^\pi \le( (\pr_\varth S)^2 + \frac{(4m\cos\varth+(4+m^2)\cos^2\varth)}{\sin^2\varth} S^2\ri)\sin\varth\,\d\varth \ri) \\
                       & \quad + \sup_{\norm{S}=1}\le(\int_0^\pi (\widetilde{F}(\varth)-a^2k^2\cos(2\varth))S^2\,\sin\varth\d\varth\ri) \\
    & \leq (|m|-4) + \sup_{\varth=0\cdots\pi} (\widetilde{F}(\varth)-a^2k^2\cos(2\varth)),
  \end{align*}
  where in the last line we used a comparison with spin-weighted harmonics in the $a=0$ case. If the Hawking-Reall bound is violated, \emph{i.e.} $\de>1$, we have
  \begin{align*}
    & \; \sup_{\varth=0\cdots\pi} (\widetilde{F}(\varth)-a^2k^2\cos(2\varth)) \\
    = & \; \sup_{\varth=0\cdots\pi} \le(\frac{\Xi}{\Deth}X(1+X)4m\cos\varth + \frac{\Xi}{\Deth}(1 -\de)(1-X)\le((1+X +\de (1 - X)) m^2-4mX\cos\varth\ri)-X^2\cos(2\varth)\ri) \\
    \leq & \; 4|m| X(1+X) + 4|m|X(1-X)(\de-1) + (1-\de)\Xi(1-X)m^2 + X^2.
  \end{align*}
  Thus, using that $\de>1$, we infer
  \begin{align}\label{est:limitlambdainfinity}
    \begin{aligned}
    \lat_{mm}^{m\om_+} & \leq (|m|-4) + 4|m| X(1+X) + 4|m|X(1-X)(\de-1) + (1-\de)\Xi(1-X)m^2 + X^2 \\
    & \xrightarrow{|m|\to+\infty} - \infty.
    \end{aligned}
  \end{align}
  Since $\lat_{mm}^{m\om_+}$ is the limit of the stationary potential $V_{\mathrm{stat}}^m[\lat_{mm}^{m\om_+}]$ at infinity (see Remark~\ref{rem:angularestimatessharpde1}), we deduce that -- if the Hawking-Reall bound is violated -- the (limit of the) potential can be made arbitrarily negative provided that $|m|$ is sufficiently large. See also Remark~\ref{rem:sharpHR}. We emphasise that~\eqref{est:limitlambdainfinity} holds \emph{for all admissible parameters $0 \leq X <1$.} 
\end{remark}

\subsection{Proof of Proposition~\ref{prop:statm0}}\label{sec:proofpropstatm0}
From the energy identity of Lemma~\ref{lem:energy} and the estimates on the angular eigenvalues~\eqref{est:lam0}, the proof of Proposition~\ref{prop:statm0} follows provided that we can show that for all admissible Kerr-adS black hole parameters $(M,a,k)$ we have
\begin{align}\label{eq:positivityV0}
  V_{\mathrm{stat}}^{m=0}[\lat=0](r) > 0, && \text{for all $r > r_+$.}
\end{align}

Under the rescaling
\begin{align}\label{eq:rescaling}
  r & \to M r, & a & \to Ma, & k & \to k/M, 
\end{align}
we have
\begin{align*}
  V_{\mathrm{stat}}^{m=0}[\lat=0](r) & \to \frac{1}{M^2} \le(V_{\mathrm{stat}}^{m=0}[\lat=0](r)\ri)\bigg|_{M=1}
\end{align*}
and the proof of~\eqref{eq:positivityV0} reduces to the $M=1$ case.\footnote{More generally, using the change of coordinates $t\to M t$, $r\to Mr$, we have that $g_{\mathrm{KadS}}\big|_{(M,a,k)} = M^2 g_{\mathrm{KadS}}\big|_{(1,a/M,Mk)}$. Thus, we can always restrict to study Kerr-adS spacetimes with $M=1$.} Let us consider the $a=0$ case. For all $r\geq r_+$ we have
\begin{align*}
  (r^2+a^2)^4 V_{\mathrm{stat}}^{m=0}[\lat=0](r) & = r^5\le(-6 + 4r+18k^2r^2\ri) \\
                                                 & \geq r_+^5\le(-6 + 4r_++18k^2r_+^2\ri) \\
                                                 & \geq r_+^5\le(-6 + 4r_+ + 18\le(\frac{2}{r_+}-1\ri)\ri)\\
                                                 & > 0,
\end{align*}
where in the last inequality we used that $0 < r_+ < 2$.\\

Let us now consider the case $a>0$. Let us define
\begin{align*}
  x & := r/a, & x_+ & := r_+/a, & X & := ak.
\end{align*}
From $\De(r_+)=0$ we deduce
\begin{align}\label{eq:defKK}
  k & = \KK(x_+,X) := \frac{X(1+x_+^2)(1+x_+^2X^2)}{2x_+}, & a & = X/k. 
\end{align}
We can therefore swap the two black hole parameters $(a,k)$ for the parameters $(x_+,X)$. To determine the domain of the parameters $(x_+,X)$, let define 
\begin{align*}
  \KK'(x_+,X) & := -1+3x_+^4X^2+x_+^2(1+X^2).
\end{align*}
We have the following lemma.
\begin{lemma}\label{lem:globdiffeoparameters}
  Denote by $\mathcal{A}$ the set of admissible subextremal Kerr-adS black hole parameters $(M,a,k)$ of Definition~\ref{def:admissibleKadSparameters}. The map
  \begin{align*}
    \theta:(x_+,X) \mapsto \le(a=X/\KK(x_+,X),k=\KK(x_+,X)\ri)
  \end{align*}
  is a bijection from $\BB:=\{\KK'\geq 0\}\cap\{0 < X < 1\}$ onto the set of admissible Kerr-adS black hole parameters $\AA_1:= \{(a,k):(1,a,k) \in\AA,~a\neq 0\}$. We call $\BB$ the set of \emph{admissible parameters $(x_+,X)$}. See Figure~\ref{fig:arplus} for a graphical representation of $\BB$. 
\end{lemma}
\begin{proof}
  By the definition of admissible parameters and the above definitions of $x_+$ and $X$, we easily have
  \begin{align*}
    \AA_1 = \theta\le(\{0 < X <1\}\cap \{x_+>0\} \ri).
  \end{align*}
  Let $0<X<1$ be fixed. A direct computation gives
  \begin{align}\label{eq:derivKK}
    \pr_{x_+}\KK(x_+,X) & = \frac{X}{2x_+^2}\KK'(x_+,X).
  \end{align}
  The map $x_+\mapsto \KK'(x_+,X)$ is strictly increasing and $\KK'(0,X)=-1<0$ thus there exists a unique $x_c(X)>0$ such that $\KK'(x_+,X) > 0$ for all $x_+>x_c(X)$ and $\KK'(x_+,X)< 0$ for all $x_+< x_c(X)$. Therefore, the map $\phi:x_+\mapsto \KK(x_+,X)$ strictly decreases on $(0,x_c(X))$ and strictly increases on $(x_c(X),+\infty)$. Since $\phi(x_+) \to +\infty$, we deduce that $\phi$ is a bijection from $[x_c(X),+\infty)$ onto $\phi\le((0,+\infty)\ri)$ and this finishes the proof of the lemma. 

  
  \begin{figure}[h!]
    \centering
    \includegraphics[width=0.4\textwidth,height=0.4\textwidth]{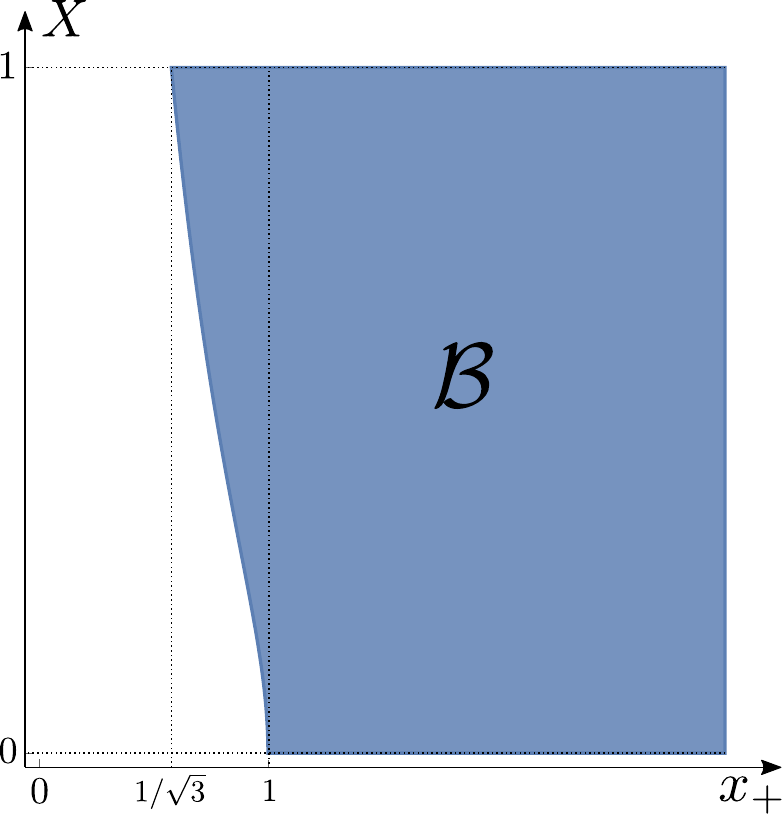}
    \caption{The set of admissible parameters $(x_+,X)$.}
    \label{fig:arplus}
  \end{figure}
\end{proof}

Using the new variables $X=ak,x=r/a$, we rewrite
\begin{align}\label{eq:Vtilde0}
  \begin{aligned}
    a^{-4}k^{2}(r^2+a^2)^4V_{\mathrm{stat}}^{m=0}[\lat=0](r) & = 4 k^2 (1 + 4 x^2) + (1 + x^2)^2 X^2 (1 - X^2) \le(1+x^2 (4-3X^2)\ri) \\
    & \quad + 2 k x (1 + x^2) X \le(9 x^4 X^2 - 7-3X^2 + x^2 (-3 + 2 X^2)\ri) \\
    & =: \widetilde{V}_0,
  \end{aligned}
\end{align}
which is a polynomial of degree $7$ in $x$. The strict positivity of $\widetilde{V}_0(x)$ for $x>x_+$ will follow if
\begin{align}
  \label{eq:derivV0}
  \pr_x^{j}\widetilde{V}_{0}(k=\KK(x_+,X),x=x_+) \geq 0, && \text{for all admissible parameters $(x_+,X)$,}
\end{align}
for all $0\leq j \leq 7$ and if at least one of these quantities is strictly positive. It is clear from the expression~\eqref{eq:Vtilde0} of $\widetilde{V}_0$ that $\pr^{6}_x\widetilde{V}_0,\pr^7_x\widetilde{V}_0$ are strictly positive if $0<X<1$. Using~\eqref{eq:defKK}, a computation for the lower derivatives gives
\begin{align*}
  \widetilde{V}_0(x=x_+) & = x_+^{-2}X^2(1+x_+^2)^2\le(-1 + 3 x_+^4 X^2 + x_+^2 (1 + X^2)\ri)^2, \\ \\
  \pr_x\widetilde{V}_0(x=x_+) & = x_+^{-1}(1+x_+^2)X^2\big(-1 + 3 x_+^4 X^2 + x_+^2 (1 + X^2)\big) \big(-1 + 9 x_+^2 + 3 X^2 + 14 x_+^2 X^2 + 21 x_+^4 X^2\big),\\ \\
  \pr_x^2\widetilde{V}_0(x=x_+) & = 2x_+^{-2}X^2\big(4 - 16 x_+^2 - 2 x_+^4 + 30 x_+^6 - 4 x_+^2 X^2 + 3 x_+^4 X^2 + 142 x_+^6 X^2 \\
                         & \quad +  159 x_+^8 X^2 + 3 x_+^2 X^4 + 37 x_+^4 X^4 + 160 x_+^6 X^4 + 303 x_+^8 X^4 +   189 x_+^{10} X^4\big),\\ \\
  \pr_x^3\widetilde{V}_0(x=x_+) & = 6x_+^{-1}X^2\big(-10 - 4 x_+^2 + 50 x_+^4 - X^2 + 39 x_+^2 X^2 + 245 x_+^4 X^2 +  285 x_+^6 X^2 \\
                         & \quad + 23 x_+^2 X^4 + 169 x_+^4 X^4 + 425 x_+^6 X^4 + 315 x_+^8 X^4\big), \\ \\
  \pr_x^4\widetilde{V}_0(x=x_+) & = 24 X^2  \big(-6 + 45 x_+^2 + 40 X^2 + 250 x_+^2 X^2 + 300 x_+^4 X^2 + 6 X^4 + 100 x_+^2 X^4 + 370 x_+^4 X^4 + 315 x_+^6 X^4\big),\\ \\
    \pr_x^5\widetilde{V}_0(x=x_+) & = 120 x_+^{-1} X^2 \big(-3 + 21 x_+^2 + 11 X^2 + 155 x_+^2 X^2 + 186 x_+^4 X^2 + 29 x_+^2 X^4 + 200 x_+^4 X^4 + 189 x_+^6 X^4\big).
\end{align*}

It is immediate that $\widetilde{V}_0(x=x_+)$ is non-negative. Using that $x_+\geq 1/\sqrt{3}$, we easily obtain that $\pr_x\widetilde{V}_0(x=x_+),\pr_x^4\widetilde{V}_0(x=x_+),\pr_x^5\widetilde{V}_0(x=x_+)$ are non-negative for all admissible parameters $(x_+,X)$. Let us define
\begin{align*}
  P_2(x_+,X) & := 4 - 16 x_+^2 - 2 x_+^4 + 30 x_+^6 - 4 x_+^2 X^2 + 3 x_+^4 X^2 + 142 x_+^6 X^2 \\
             & \quad +  159 x_+^8 X^2 + 3 x_+^2 X^4 + 37 x_+^4 X^4 + 160 x_+^6 X^4 + 303 x_+^8 X^4 +   189 x_+^{10} X^4.
\end{align*}
For all admissible parameters $(x_+,X)$ we have the following sequence of inequalities
\begin{align*}
  P_2(x_+,X) & \geq 4 - 16 x_+^2 - 2 x_+^4 + 30 x_+^6 + \big(- 4 x_+^2 + 3 x_+^4 + 142 x_+^6 \big)X^2 \\
             & \geq 4 - 16 x_+^2 - 2 x_+^4 + 30 x_+^6 + \frac{115}{27}X^2 \\
             & \geq 4 - 16 x_+^2 - 2 x_+^4 + 30 x_+^6 + \frac{115}{27} \le(\frac{1-x_+^2}{x_+^2+3x_+^4}\ri) \\
             & \geq 0.
\end{align*}
Thus, $\pr_x^2\widetilde{V}_0(x=x_+) = 2x_+^{-2}X^2P_2(x_+,X)$ is non-negative. Arguing along the same lines we obtain that $\pr_x^3\widetilde{V}_0(x=x_+)$ is non-negative. This finishes the proof of~\eqref{eq:positivityV0} and of Proposition~\ref{prop:statm0}.

\subsection{Proof of Proposition~\ref{prop:stat}}\label{sec:proofpropstat}
This section is dedicated to the proof of Proposition~\ref{prop:stat}. We first treat in Sections~\ref{sec:proofpropstat1} and~\ref{sec:proofpropstat2} two special cases for which we have positivity of the angular eigenvalue (see Section~\ref{sec:angularestimates}). For these cases it suffices to use the positivity of the potential obtained in Section~\ref{sec:proofpropstatm0}. In Section~\ref{sec:proofpropstat3} we treat the remaining cases where no positive bounds for the angular eigenvalue is available. To conclude, we use the Hardy estimate and Hardy potential of Section~\ref{sec:hardy} and follow similar lines as in the proof of~\eqref{eq:positivityV0}.  

\subsubsection{The~\eqref{est:assumboundX}, $|m|=1$ case}\label{sec:proofpropstat1}
Let us first treat the~\eqref{est:assumboundX}, $|m|=1$ case. Using the positivity of the potential~\eqref{eq:positivityV0} obtained in the previous section, we have
\begin{align*}
  V^{m=\pm1}_{\mathrm{stat}}[\lat=0](r) & = \le(\frac{\Xi\om_+}{k}\ri)^2\frac{\De-k^2(r-r_+)^2(r+r_+)^2}{(r^2+a^2)^2} + V^{m=0}_{\mathrm{stat}}[\lat=0](r) \geq V^{m=0}_{\mathrm{stat}}[\lat=0](r) > 0.
\end{align*}
Plugging this estimate and the positivity of the angular eigenvalue~\eqref{est:lampos} into the energy identity~\eqref{est:energy}, we deduce Proposition~\ref{prop:stat} in the~\eqref{est:assumboundX} case in for $|m|=1$.

\subsubsection{The $a=0$ case}\label{sec:proofpropstat2}
From Lemma~\ref{lem:angulardecompo}, and the fact that $\om_+=a/(r_+^2+a^2) = 0$, we have
\begin{align*}
  \lat_{m\ell}^0 & = \la^{0}_{m\ell}-2 = \ell(\ell+1) -4 \geq 0. 
\end{align*}
since $\ell\geq 2$. Moreover, using that $\om_+=0$ and the positivity~\eqref{eq:positivityV0}, we have
\begin{align*}
  V^{m}_{\mathrm{stat}}[\lat=0](r) & = V^{m=0}_{\mathrm{stat}}[\lat=0](r) \geq V^{m=0}_{\mathrm{stat}}[\lat=0](r) > 0.
\end{align*}
As in the previous section, this finishes the proof of Proposition~\ref{prop:stat} if $a=0$. 

\subsubsection{The remaining cases}\label{sec:proofpropstat3}
From the energy identity of Lemma~\ref{lem:energy} -- using that for $|m|\geq 2$, $V^{m}_{\mathrm{stat}} \geq V^{m=2}_{\mathrm{stat}}$ --, the Hardy estimate of Lemma~\ref{lem:energyandHardy} and the estimates on the angular eigenvalues of Lemma~\ref{lem:lambda}, the proof of Proposition~\ref{prop:stat} in the remaining cases follows from the following lemma.
\begin{lemma}
  Let $a>0$. If the Hawking-Reall bound~\eqref{est:HRboundoriginal} is satisfied, we have 
  \begin{align}\label{est:Vtilde2pos}
    \begin{aligned}
      V_{-5\de/2} & := V^{m=2}_{\mathrm{stat}}[\lat=-5\de/2](r) + V_{\mathrm{Hardy}}[\lat=-5\de/2](r) >0,
    \end{aligned}
  \end{align}
  for all $r>r_+$. Moreover, if the bound~\eqref{est:farHR} is satisfied, we have 
  \begin{align}\label{est:Vposfarm1}
    \begin{aligned}
      V_{\mathrm{far},1} & := V^{m=1}_{\mathrm{stat}}[\lat=-(1+X)^2 -2X(1-X)](r) \\
      & \quad + V_{\mathrm{Hardy}}[\lat=-(1+X)^2 -2X(1-X)](r)\\
      & >0,
    \end{aligned}
  \end{align}
  and
  \begin{align}\label{est:Vposfarm2}
    \begin{aligned}
      V_{\mathrm{far},2} & := V^{m=2}_{\mathrm{stat}}[\lat=-4(1+X)^2 + 4\Xi(1-\de^2)(1-X)^2](r) \\
      & \quad + V_{\mathrm{Hardy}}[\lat=-4(1+X)^2 + 4\Xi(1-\de^2)(1-X)^2](r)\\
      & >0,
    \end{aligned}
  \end{align}
  for all $r>r_+$. 
\end{lemma}
\begin{proof}
  The proofs of~\eqref{est:Vtilde2pos} and~\eqref{est:Vposfarm1} go along the same lines (and are easier) as the proof of~\eqref{est:Vposfarm2}, and we refer the reader to the Mathematica notebook for the specific computations. Using the rescaling~\eqref{eq:rescaling}, the proof of~\eqref{est:Vposfarm2} reduces to the $M=1$ case. We refer to Section~\ref{sec:proofpropstatm0} for the definitions of the parameters and variables used in this proof. From an inspection of the expression of $V_{\mathrm{far},2}$, the quantity
  \begin{align*}
    \widetilde{V}_{\mathrm{far},2} & := \le(\frac{X^2(1+x_+^2)^8(1+X^2x_+^2)^6}{16x_+^4}\ri)(r^2+a^2)^4V_{\mathrm{far},2}
  \end{align*}
  is a polynomial of order $8$ in $x=r/a$, and~\eqref{est:Vposfarm2} follows provided that for $0\leq i \leq 8$ we have $\pr_x^i\widetilde{V}_{\mathrm{far},2}(x=x_+) \geq 0$ (and that at least one is strictly positive). Using formula~\eqref{eq:defKK} for $k$ and that, by definition, we have $\de = \frac{1+X}{X(1+x_+^2)}$, a direct computation gives for the first coefficient 
  \begin{align*}
    \widetilde{V}_{\mathrm{far},2}(x=x_+) & = 4 X^2 (1+x_+^2)^4 \big(-1+(1+X^2) x_+^2+3 X^2 x_+^4\big)^2 \geq 0.
  \end{align*}
  A computation of the next coefficient gives
  \begin{align*}
    \pr_x\widetilde{V}_{\mathrm{far},2}(x=x_+) & = 4 x_+ (1 + x_+^2)^2X^{-1}\le(\sum_{i=0}^5 P_i(X)\le(X x_+^2-4\ri)^i\ri),
  \end{align*}
  where
  \begin{align*}
    P_0(X) & := 560 + 10036 X + 41932 X^2 + 13345 X^3 + 11380 X^4 - 455 X^5 - 2356 X^6 - 192 X^7,\\
    P_1(X) & := 428 + 10648 X + 57683 X^2 + 16376 X^3 + 12181 X^4 - 1472 X^5 - 2057 X^6 - 112 X^7,\\
    P_2(X) & := 108 + 4211 X + 31540 X^2 + 7323 X^3 + 4728 X^4 - 1042 X^5 - 592 X^6 - 16 X^7,\\
    P_3(X) & := 9 + 736 X + 8583 X^2 + 1424 X^3 + 790 X^4 - 272 X^5 - 56 X^6,\\
    P_4(X) & := 48 X + 1164 X^2 + 102 X^3 + 48 X^4 - 24 X^5,\\
    P_5(X) & := 63X^2.
  \end{align*}
  We easily check that $P_{i}(X) > 0$ for all $0<X<1$ and all $0\leq i \leq 5$. Thus, using that the bound~\eqref{est:farHR} rewrites $Xx_+^2-4\geq 0$, we have that $\pr_x\widetilde{V}_{\mathrm{far},2}(x=x_+)>0$. Along the same lines, we check that all the coefficients $\pr_x^i\widetilde{V}_{\mathrm{far},2}(x=x_+)$ are positive for $2\leq i \leq 8$ (the computations are displayed in the companion Mathematica file). This finishes the proof of the lemma.
\end{proof}

\appendix
\section{An alternative proof of the $\om-m\om_+\neq0$ case via Robin boundary conditions}\label{sec:Robin}
The following proposition shows that if a mode solution is regular at the horizon in the sense of~\eqref{eq:radialregconditions}, then the conformal anti-de Sitter boundary conditions~\eqref{eq:defbdycond} are equivalent to Robin boundary conditions together with the requirement that the coupled radial Teukolsky quantity is the (renormalised, complex conjugate of the) Teukolsky-Starobinsky transformation. These Robin boundary conditions were first derived in~\cite{Dia.San13}. 
\begin{proposition}[Robin boundary conditions]\label{prop:RobinBC}
  Let $m\in\ZZZ$, $\om\in\RRR$ and $\ell \geq |m|$ and $R_{[\pm2],\om}^{m\ell}:(r_+,+\infty)_r\to\CCC$ be two smooth functions. The following three items are equivalent.
  \begin{enumerate}
  \item\label{item:BCadS} $R_{[\pm2],\om}^{m\ell}$ satisfy the Teukolsky equations~\eqref{eq:radialTeuk}, the regularity conditions~\eqref{eq:radialreggeneral} at the horizon and the conformal anti-de Sitter boundary conditions~\eqref{eq:radialbdyconditions}.
  \item\label{item:Robin+2} $R_{[+2],\om}^{m\ell}$ satisfies the Teukolsky equation~\eqref{eq:radialTeuk}, the regularity conditions~\eqref{eq:radialreggeneral} at the horizon and the following two (linearly dependent) \emph{Robin boundary conditions}
    \begin{align}\label{eq:RobinBC}
      \begin{aligned}
        (\wp_0 + Z +12iM\Xi\om) R_{[+2],\om}^{m\ell} -  i \wp_1 \pr_{r^\ast}R_{[+2],\om}^{m\ell} & \xrightarrow{r\to+\infty} 0,\\ \\
         i\wp_2R_{[+2],\om}^{m\ell}  + (-\wp_0 + Z +12iM\Xi\om) \pr_{r^\ast}R_{[+2],\om}^{m\ell} & \xrightarrow{r\to+\infty} 0,
      \end{aligned}
    \end{align}
    where the coefficients $\wp_0,\wp_1,\wp_2$ are the transition coefficients introduced in Lemma~\ref{lem:TStransmissioninfinity}, and we have
    \begin{align}\label{eq:TeukidZ}
      R_{[-2],\om}^{m\ell} & = \le(Z^{-1}R_{[-2],\om,c}^{m\ell}\ri)^\ast,
    \end{align}
    where
    \begin{align}\label{eq:defZ}
      Z = Z_\pm(m,\om,\la_{m\ell}^\om) := \pm \sqrt{\aleph(m,\om,\la^\om_{m\ell})} - 12 i M \Xi\om.
    \end{align}
  \item\label{item:Robin-2} Item~\ref{item:Robin+2} holds with $R_{[\pm2]}$ replaced by $R_{[\mp2]}$.
  \end{enumerate}
\end{proposition}
\begin{proof}
  Assume that Item~\ref{item:BCadS} holds. The asymptotics~\eqref{eq:radialregconditionsTSgeneral} of Lemma~\ref{lem:TShorasympt} for the Teukolsky-Starobinsky transformations together with the regularity conditions~\eqref{eq:radialreggeneral} imply that there exists $Z\in\CCC^\ast$ such that
  \begin{align}\label{eq:pfRobin1}
    R_{[-2]} & = \le(Z^{-1}R_{[-2],c}\ri)^\ast.
  \end{align}
  Using the conformal anti-de Sitter boundary conditions~\eqref{eq:radialbdyconditions} and relation~\eqref{eq:pfRobin1} we infer
  \begin{align}\label{eq:RobinBCus}
    \begin{aligned}
      & R_{[+2]} - R_{[-2]}^\ast  = R_{[+2]} - Z^{-1} R_{[-2],c} \to 0,\\
      & \pr_{r^\ast}R_{[+2]} + \pr_{r^\ast}R_{[-2]}^\ast = \pr_{r^\ast}R_{[+2]} + Z^{-1}\pr_{r^\ast}R_{[-2],c} \to 0,
    \end{aligned}
  \end{align}
  when $r\to+\infty$. The Robin boundary conditions~\eqref{eq:RobinBC} follow from~\eqref{eq:RobinBCus} and the limits~\eqref{eq:matrixlimits} for the Teukolsky-Starobinsky transformations. The Robin boundary conditions lead to non-trivial solutions only if they are linearly dependent. A direct computation shows that this only occurs provided that $Z$ takes one of two values~\eqref{eq:defZ}. This finishes the proof of Item~\ref{item:Robin+2}. That Item~\ref{item:Robin+2} implies Item~\ref{item:BCadS} is obtained by rewinding the above argument and using the results of Lemmas~\ref{lem:radTS} and~\ref{lem:TShorasympt} for Teukolsky-Starobinsky transformations. The equivalence between Item~\ref{item:BCadS} and Item~\ref{item:Robin-2} is obtained along the exact same lines replacing $R_{[\pm2]}$ by $R_{[\mp2]}$.  
\end{proof}

From the Robin conditions of Proposition~\ref{prop:RobinBC} we infer an alternative proof of the main theorem in the non-stationary case.
\begin{proof}[Alternative proof of Theorem~\ref{thm:main} in the $\om-m\om_+\neq0$ case]
  Define the Wronskian $W = W(R_{[+2]},R_{[-2]})$. The functions $R_{[+2]}$ and $R_{[-2]}$ satisfy the same second order ODE (see Equations~\eqref{eq:radialTeuk} and Lemma~\ref{lem:Teukeqrel})
  \begin{align*}
    0 & = \pr_{r^\ast}^2R +\pr_{r^\ast}\le(\log\le(\frac{(r^2+a^2)^4}{\De^2}\ri)\ri)\pr_{r^\ast}R + V^{m,\om}[\la^\om_{m\ell}]R,
  \end{align*}
  where $R=R_{[\pm2]}$. The Wronskian $W$ thus satisfies the first order ODE
  \begin{align*}
    0 & = \pr_{r^\ast}W + \pr_{r^\ast}\le(\log\le(\frac{(r^2+a^2)^4}{\De^2}\ri)\ri)W,
  \end{align*}
  and we infer that
  \begin{align}\label{eq:WsolODE}
    W(r) & = \frac{\De^2}{k^4(r^2+a^2)^4}\le(\lim_{r\to+\infty} W(r)\ri).
  \end{align}
  Using the conformally anti-de Sitter conditions~\eqref{eq:radialbdyconditions}, we have
  \begin{align}\label{eq:WsolODE2}
    \begin{aligned}
      \lim_{r\to+\infty} W(r) & = \lim_{r\to+\infty} \le(R_{[+2]}\pr_{r^\ast}R_{[-2]} - R_{[-2]}\pr_{r^\ast}R_{[+2]}\ri) \\
      & = -\lim_{r\to+\infty} \le(R_{[+2]}\pr_{r^\ast}R_{[+2]}^\ast + R_{[+2]}^\ast\pr_{r^\ast}R_{[+2]}\ri).
      \end{aligned}
  \end{align}
  Using that $\aleph>0$ by Lemma~\ref{lem:angTS}, we have
  \begin{align}\label{eq:conditionfinale}
    (\wp_0,Z_{\pm}+12iM\Xi\om,\wp_1,\wp_2) \neq (0,0,0,0).
  \end{align}
  Thus, at least one of the Robin boundary conditions~\eqref{eq:RobinBC} for $R_{[+2]}$ is non-trivial. This implies that either
  \begin{align*}
    R_{[+2]} & \to 0, && \text{or} & \pr_{r^\ast}R_{[+2]} + i \sigma R_{[+2]} & \to 0,
  \end{align*}
  when $r\to+\infty$ and with $\sigma\in\RRR$. Plugging these in~\eqref{eq:WsolODE},~\eqref{eq:WsolODE2}, we obtain that $W=0$ and that $R_{[+2]}$ and $R_{[-2]}$ are proportional. If $\om-m\om_+\neq0$, from the asymptotics conditions at the horizon~\eqref{eq:radialreggeneral}, this implies that $R_{[+2]}=R_{[-2]}=0$ which proves Theorem~\ref{thm:main} in the non-stationary case.
\end{proof}

\begin{remark}\label{rem:TSpositivitycriticalbis}
  As in the first proof of Section~\ref{sec:TSconslaws}, we crucially use the non-vanishing of the angular Teukolsky-Starobinsky constant $\aleph$. Indeed, we emphasise that the conditions~\eqref{eq:conditionfinale} appearing in the above proof are actually almost equivalent to the non-vanishing of $\aleph$. First, given the definitions of $\wp_0,\wp_1,\wp_2$ and $\aleph$ (or also using the radial Teukolsky-Starobinsky inversion formulas and~\eqref{eq:matrixlimits}), it is easy to show that $\wp_0^2 + \wp_1\wp_2 = \aleph$. It is also easy to check from their definitions that $\wp_1=0$ implies $\wp_2=0$. Thus, the conditions $(\wp_0,Z_{\pm}+12iM\Xi\om,\wp_1,\wp_2)=(0,0,0,0)$ actually reduce to $\wp_1=0$ and $\aleph=0$. Moreover, the vanishing of $\wp_1$ itself is not independent of the vanishing of $\aleph$: one can check that $\wp_1=0$ when the triplets $(m,\la,\Xi\om)$ take the seven special values~\eqref{eq:m0vanishingTS},~\eqref{eq:m1vanishingTS} of the proof of Lemma~\ref{lem:angTS}.
\end{remark}

\bibliographystyle{graf_GR_alpha}
\bibliography{graf_GR}
\end{document}